\documentclass[12pt]{article}
\usepackage[utf8]{inputenc}
\usepackage{graphicx,psfrag,epsf}
\usepackage{booktabs}
\usepackage{textgreek}
\usepackage{threeparttable}
\usepackage{float}
\usepackage{amsmath,amsfonts,amsthm,bm} 
\usepackage{url}
\usepackage{transparent}
\usepackage{adjustbox}
\usepackage{flafter}
\usepackage{lscape}
\usepackage{caption}
\usepackage{subcaption}
\usepackage{graphicx}
\usepackage{geometry}
\usepackage{footnote} 
\makesavenoteenv{tabular} 
\usepackage{verbatim}
\usepackage{natbib}
\usepackage{comment}
\usepackage{rotating}
\usepackage{hyperref}
\usepackage{mdframed}
\usepackage{lipsum}
\usepackage{array}
\newcolumntype{H}{>{\setbox0=\hbox\bgroup}c<{\egroup}@{}}
\usepackage{xcolor}
\usepackage{amsmath}
\usepackage{multirow}
\usepackage{bbm}
\usepackage{algorithmicx}
\usepackage{algorithm,algpseudocode}
\usepackage{longtable}

\newcommand{\blind}{0}
\newtheorem{theorem}{Theorem}
\newtheorem{lemma}{Lemma}

\begin{document}

\def\spacingset#1{\renewcommand{\baselinestretch}%
{#1}\small\normalsize} \spacingset{1}

\def\spacingset#1{\renewcommand{\baselinestretch}%
{#1}\small\normalsize} \spacingset{1}


\if0\blind
{
  \title{\bf Fused LASSO as Non-Crossing Quantile Regression\thanks{The authors thank Atanas Christev, Rod McCrorie, Paul Allanson, Isaiah Andrews, Istv\'an J\'ar\'asi, Katalin Varga, David Kohns, and all the participants of the 2022 and 2023 PhD conference in Crieff for their feedback. Tibor Szendrei thanks the ESRC for PhD studentship as well as Heriot-Watt University for institutional support. The usual disclaimer applies.}}
  \author{Tibor Szendrei \footnote{Corresponding author: t.szendrei@niesr.ac.uk}\\
    Department of Economics, Heriot-Watt University, UK. \\
    National Institute of Economic and Social Research, UK.\\
     \\
    Arnab Bhattacharjee\\
    Department of Economics, Heriot-Watt University, UK.\\
    National Institute of Economic and Social Research, UK.\\
    \\
    Mark E. Schaffer\\
    Department of Economics, Heriot-Watt University, UK.}
  \maketitle
} \fi

\if1\blind
{
  \bigskip
  \bigskip
  \bigskip
  \begin{center}
    {\LARGE\bf Fused LASSO as Non-Crossing Quantile Regression}
\end{center}
  \medskip
} \fi


\bigskip
\begin{abstract}
\noindent Growth-at-Risk is vital for empirical macroeconomics but is often suspect to quantile crossing due to data limitations. While existing literature addresses this through post-processing of the fitted quantiles, these methods do not correct  the estimated coefficients. We advocate for imposing non-crossing constraints during estimation and demonstrate their equivalence to fused LASSO with quantile-specific shrinkage parameters. By re-examining Growth-at-Risk through an interquantile shrinkage lens, we achieve improved left-tail forecasts and better identification of variables that drive quantile variation. We show that these improvements have ramifications for policy tools such as Expected Shortfall and Quantile Local Projections.
\end{abstract}


\noindent%
{\it Keywords:} Interquantile shrinkage, Crossing Quantile Curves, High-Dimensional Econometrics, Growth-at-Risk. \\
\noindent
\vfill

\spacingset{1.45} 


\section{Introduction}

Growth-at-Risk (GaR) has become a key measure of economic vulnerability since the work of \citet{adrian2019vulnerable}. The key idea is to think about GDP growth through the lens of value-at-risk, and use quantile regression of \citet{koenker1978regression} to uncover nonlineary macro-financial linkages. The need for GaR was highlighted by the global financial crisis which showed how downside risks, or lower quantiles of the density of GDP growth, evolve with the state of financial market. 

GaR has become a popular tool for policy makers and researchers alike. \citet{figueres2020vulnerable} apply the method to euro area data and show that the macro-financial linkages driving skewness are not specific to the US. \citet{kohns23hsbqr} use shrinkage methods to fit GDP densities with over 200 covariates and find that accounting for multiple sources of risk leads to better density fits of GDP growth. \citet{iseringhausen2023aggregate} use the same dataset and look at the skewness of the fitted quantile estimates, which encompasses the aggregate impact of the considered covariates. They show that aggregate skewness is highly pro-cyclical. \citet{mitchell2022constructing} use quantiles estimated on US GaR and show that multimodality in the fitted densities is just as important as skewness.\footnote{\citet{adrian2021multimodality} and \citet{kohns23hsbqr} also discuss the importance of multimodality stemming from macro-financial linkages}. Finally, \citet{carriero2025specification} has looked at specifications for quantile regression for empirical macroeconomics, with GaR being one of the models being investigated.

While GaR is undoubtedly an important policy tool, it is a macroeconometric approach which in turn has data constraints. This data scarcity is particularly a problem for tail estimation \citep{koenker2005}, which is the main goal for a GaR model. One big concern when applying quantile regression in data scarce settings is the notion of quantile crossing, i.e. the situation where the fitted quantiles are not monotonically increasing. This yields improper densities and the literature has tackled this problem using one of two methods. One approach is to fit some distribution one the estimated quantiles in each time period, as done in \citet{adrian2019vulnerable} or \citet{korobilis2017quantile}. Another approach is to simply sort the estimated quantiles in each period as advocated by \citet{chernozhukov2009improving} and \citet{chernozhukov2010crossing}.\footnote{There is also the novel method of \citet{mitchell2022constructing} which proposes a nonparametric method to construct densities from sorted quantiles.} 

All these approaches that have been popular in the macroeconometric literature involve some form of post-processing the estimated quantiles. While these two-step methods yield proper densities, simply correcting the fitted quantiles does not quantify corresponding changes in the estimated coefficients, i.e., the coefficients estimated in the first step remain uncorrected. This is particularly a problem as quantile coefficients are used to construct quantile IRFS \citep{chavleishvili2024forecasting}, or quantile local projections \citep{ruzicka2021quantile}. As such, for macroeconometric inference it would be desirable to impose non-crossing during estimation of the quantiles. This is the motivation in \citet{szendrei2023revisiting}, who use the non-crossing constraints of \citet{bondell2010noncrossing} to identify variables that drive macro-financial nonlinearity in the Euro Area GaR. \citet{szendrei2023revisiting} also find that imposing non-crossing constraints also improves the forecasted densities of GaR.

While non-crossing constraints are alluring for empirical macroeconomic applications of quantile regression, we do not know how these constraints influence the estimated coefficients. The primary motivation of this paper is to study the impact of such constraints on the model parameters. This is achieved by implementing a set of non-crossing constraints that can be scaled, which in turn makes the non-crossing constraints tighter or looser. Using this adaptive non-crossing constraint we show that non-crossing constraints are equivalent to fused LASSO with quantile specific shrinkage parameters. We propose an estimator using these adaptive non-crossing constraints (Generalised Non-crossing Constraint Quantile Regression), where the hyperparameter regulating interquantile shrinkage is obtained by cross-validation.

We conduct comprehensive Monte Carlo exercises based on those used in the original paper by \citet{bondell2010noncrossing}, that first proposed non-crossing quantile regression. Through these experiments we show that the proposed estimator is capable of providing further improvements in fit compared to the original BRW estimator. We also consider the variable selection properties of the different estimators and verify that non-crossing constraints are a type of fused shrinkage. The Monte Carlos also reveal how the BRW estimator undershrinks quantile variation while the traditional Fused LASSO, as outlined in \citet{jiang2013interquantile}, overshrinks.

We then estimate US GaR using the variables outlined in \citet{adrian2019vulnerable}. We will investigate the canonical GaR with a interquantile shrinkage framework, i.e. identifying variables that drive nonlinearities. \citet{adrian2019vulnerable} alluded to this in their paper, but the estimator they use did not enforce limiting quantile variation during estimation, which influences their downside risk measures. In this GaR exercise our proposed estimator fares much better at distinguishing quantile varying variables than the traditional Fused LASSO. Furthermore, in a small pseudo forecasting exercise we find that the estimator with adaptive non-crossing constraints yields the best left tail forecasts. These findings add to \citet{carriero2025specification}, namely that fused shrinkage is also important for GaR. Our results also indicate that it is also important how this shrinkage is induced: adaptive non-crossing constraints yield improvements, while the traditional fused LASSO seem to lead to worse forecasting performance.

These differences in GaR coefficients have an influence on tools used in policy. We show how the choice of imposing interquantile shrinkage, and how to impose it, influences the estimated Expected Shortfall of GDP growth. In particular, Expected Shortfall based on traditional quantile regression is more volatile from one period to the next, while the proposed estimator produces smoother expected shortfall values over time, while still having values during periods of financial stress. We also produce quantile local projections of how GDP responds over time to a unit shock in the financial stress measure. Just like with the expected shortfall values, the LPs of the proposed estimator are smoother than those of the quantile regression LPs. These examples show the policy relevance of implementing interquantile shrinkage in GaR.

This paper is structured as follows. In Section~2 we introduce the quantile regression of \citet{koenker1978regression} along with the non-crossing constraint of \citet{bondell2010noncrossing}, before providing the adaptive non-crossing constraints that vary with a hyperparameter $\alpha$. Using these new constraints, we show that one can rewrite them into the Fused LASSO constraint of \citet{jiang2013interquantile}. The section concludes by describing how the hyperparameter is chosen. Section~3 describes the Monte Carlo experiment before presenting the fit and variable selection performance of the different estimators. This is followed by the US Growth-at-Risk application in Section~4. The paper concludes with the key takeaways of the method in Section~5.

\section{Methodology}
\subsection{Quantile Regression}
The first building block of the proposed methodology is the quantile regression framework of \citet{koenker1978regression}. Quantile regression is a weighted version of the least absolute deviation regression, and yields lines of best fit that explain different parts of the distribution. The collection of $Q$ estimated quantiles can be used to describe the distribution of a response variable $Y$ conditional on a vector of response variables $X \in \mathbb{R}^{K}$. Formally, the $q^{th}$ quantile is modelled in regression setting as:

\begin{equation} \nonumber
    \mathcal{Q}_{\tau_q}=X^T\beta_{\tau_q}.
\end{equation}

The collection of the Q quantiles parameters $\beta=\{\beta_{\tau_1},\beta_{\tau_2},...,\beta_{\tau_Q} \}$ describe the conditional distribution. The goal is to estimate the vector of coefficients $\beta_{\tau_q} \in \mathbb{R}^{K}$ for all quantiles.\footnote{Note that with this notation $X$ includes a constant or intercept.}  This can be done using quantile regression:

\begin{equation}\label{eq:QR3}
\begin{split}
    \hat{\beta}&=\underset{\beta}{argmin}\sum^{Q}_{q=1}\sum^{T}_{t=1}\rho_{\tau_q}(y_t-x_t^T\beta_{\tau_q})\\
    \rho_{\tau_q}(u)&=u({\tau_q}-I(u<0))
\end{split}
\end{equation}

\noindent where the second equation is the `tick-loss' function \citep{koenker1978regression}.

Quantile regression is not the only asymmetric estimator which uses the `tick-loss' function as a way to model different parts of the distribution. \citet{newey1987asymmetric} introduced the concept of expectile regression, i.e., the `tick-loss' weight applied to the $\ell_2$ norm of the residuals. The advantage of expectile regression is that the objective function is differentiable, unlike quantile regression's. Furthermore, when setting $\tau=0.5$ in an expectile regression, one recovers the OLS estimator, i.e., the conditional mean.

Expectile regression does come with some disadvantages. First, interpretation: while conditional quantiles can be interpreted as splitting the data into $\tau$ and $1-\tau$ parts, i.e. the $\tau^{th}$ conditional quantile of the sample, the same is not true for expectiles. Instead expectiles are the $\tau^{th}$ quantile of some distribution which is related to the original distribution of $Y$ as shown in \citet{jones1994expectiles}. Second, since expectile regression is based on the $\ell_2$ norm of the residuals, outliers in the dependent variable impacts all expectiles, while outliers only exert undue influence on extreme quantiles in the case of quantile regression. As such, ``fan-shaped'' estimated quantiles are a clear indication of heteroskedasticity, while the same shape for expectiles might simply be due to outliers in the dependent variable \citep{newey1987asymmetric, sobotka2012geoadditive}. On account of these disadvantages, in this paper we will focus on quantile regression and leave for future research to explore how our findings translate to expectiles.

Equation (\ref{eq:QR3}) gives an estimate for the parameters with which a description of the conditional distribution is obtained, but it is possible that these fitted quantiles cross. Quantile crossing is a violation of monotonocity assumption and often occurs on account of data scarcity or model misspecification \citep{koenker2005}. Limits in data availability are frequently encounted in practice, particularly in time-series settings. For forecasting applications, the two main methods for addressing quantile crossing are (1) use the fitted quantiles to fit some distribution for each time period as in \citet{adrian2019vulnerable} or \citet{korobilis2017quantile}; or (2) sort the estimated quantiles in each period as proposed by \citet{chernozhukov2009improving} and \citet{chernozhukov2010crossing}. While these two-step methods yield proper densities, correcting the fitted quantiles does not quantify corresponding changes in the estimated coefficients, i.e., the coefficients estimated in the first step remain uncorrected. This prompted \citet{bondell2010noncrossing} to propose an estimator which yields no crossing for the estimated quantiles in-sample. We note that there are other means to estimate non-crossing quantiles like \citet{liu2009stepwise} who estimate the median first, and sequentially estimates further quantiles, conditional on the previously estimated quantiles not being crossed. While this method will yield non-crossing quantiles, the choice of the first estimated quantile can be arbitary. As such, in this paper we will focus exclusively on non-crossing constraints as done in \citet{bondell2010noncrossing} since elements of their constraints have been carried over to other quantile estimators (see \citet{yang2017joint} for example).


\subsection{Non-Crossing Constraints}

Non-crossing constraints incorporated into quantile regression are a way to ensure that the estimated quantiles remain monotonically increasing. Imposing them can be done via inequality constraints:

\begin{equation} \label{eq:NCQR}
    \begin{split}
        \hat{\beta}&=\underset{\beta}{argmin}\sum^{Q}_{q=1}\sum^{T}_{t=1}\rho_q(y_{t}-x_t^T\beta_{\tau_q})\\
        &s.t.~x^T\beta_{\tau_q} \geq x^T\beta_{\tau_{q-1}}
    \end{split}
\end{equation}

While conceptually simple, the number of constraints in equation (\ref{eq:NCQR}) can be rather large on account of there being $T \times(Q-1)$ inequality constraints. To address this, \citet{bondell2010noncrossing} restrict the domain of the covariates to $\mathcal{D}\in [0,1]^K$ and focus on the worst case scenario in the data,\footnote{The situation where the negative difference coefficient's ($\gamma^-_{j,\tau_q}$) corresponding variables values are 1 and the positive difference coefficients ($\gamma^+_{j,\tau_q}$) corresponding variables equal 0.} which reduces the number of constraints to $(Q-1)$.\footnote{Any domain of interest which has an affine transformation that maps to $\mathcal{D}\in [0,1]^K$ works.} This simplifies computation a great deal and enables non-crossing constraints to be included without too much additional computational cost. Because quantile regression is invariant to monotone transformations, any affine invertible transformation can be applied: to obtain the coefficients pertaining to the un-transformed data, it is enough to apply the inverse transformation to the estimated coefficients \citep{koenker2005}.

The method of \citet{bondell2010noncrossing} recasts the parameters in terms of quantile differences: $(\gamma_{0,\tau_1},...,\gamma_{K,\tau_1})^T=\beta_{\tau_1}$ and $(\gamma_{0,\tau_q},...,\gamma_{K,\tau_q})^T=\beta_{\tau_q}-\beta_{\tau_{q-1}}$ for $q=2,...,Q$. With this quantile difference reparametrisation, the constraint in equation (\ref{eq:NCQR}) becomes $x^T\gamma_{\tau_q}\geq 0$. Note that we can without further assumptions decompose the $j^{th}$ difference as $\gamma_{j,\tau_q}=\gamma^+_{j,\tau_q}-\gamma^-_{j,\tau_q}$, where $\gamma^+_{j,\tau_q}$ is its positive and $-\gamma^-_{j,\tau_q}$ its negative part. For each $\gamma_{j,\tau_q}$ both parts are non-negative and only one part is allowed to be non-zero. With this reparameterisation, along with the restriction of the data to $\mathcal{D}\in [0,1]^K$, the non-crossing constraint can be redefined as:


\begin{equation}
    \gamma_{0,\tau_{q}}\geq \sum^K_{j=1}\gamma^-_{j,\tau_q} ~ (q=2,...,Q) \label{eq:constraint}
\end{equation}

A non-crossing constraint is therefore equivalent to ensuring that the sum of negative shifts do not push the quantile below the change in intercept, which acts as a pure location shifter. \citet{bondell2010noncrossing} shows that (\ref{eq:constraint}) is a necessary as well as a sufficient condition for non-crossing quantiles.


\subsection{Adaptive Non-Crossing Constraints}

Throughout the paper we follow the framework and assumptions laid out in \citet{bondell2010noncrossing}. To derive adaptive non-crossing constraints we first need to recast the non-crossing constraints of equation (\ref{eq:NCQR}) in a way that does not restrict the domain of interest to $\mathcal{D}\in [0,1]^K$. This is provided by Lemma \ref{theorem:generalconstr} below.

\begin{lemma} \label{theorem:generalconstr}
Given that a non-crossing constraint looks at the sum of positive shifters to be larger than the sum of negative shifters, in the worst case scenario, these constraints can be reformulated as:

\begin{equation} \label{eq:constraintEXPAND}
    \gamma_{0,\tau_{q}}+\sum^K_{j=1} min(X_j)\gamma_{j,\tau_q}^+\geq \sum^K_{j=1} max(X_j)\gamma_{j,\tau_q}^-
\end{equation}

\noindent where $X_j$ is the $j^{th}$ variable in the design matrix. 
\end{lemma}
\begin{proof}

Recall the standard non-crossing constraint formulation from \citet{bondell2010noncrossing} needs that in the worst case scenario the sum of positive shifters needs to be larger than the sum of negative shifters:
\begin{equation}
    \gamma_{0,\tau_{q}}+\sum^K_{j=1} (Z_j=0) \cdot \gamma_{j,\tau_q}^+ \geq \sum^K_{j=1} (Z_j=1) \cdot \gamma_{j,\tau_q}^- \nonumber
\end{equation}

\noindent where $Z_j\in[0,1]$ is a transformed variable of $X_j$. Assume that $\max(X_j) > \min(X_j)$ for all $j \in \{1,2,...,K\}$, i.e., $Var(X_j)>0$ for all $j$. Consider a design matrix variable $X_j$ with domain $[\min(X_j), \max(X_j)]$. $X_j$ can be normalised using the min-max transformation: $Z_{t,j} = \frac{X_{t,j} - \min(X_j)}{\max(X_j) - \min(X_j)}$, where $Z_j$ is the normalised variable. We can express $X_j=min(X_j)+Z_j\cdot[max(X_j)-min(X_j)]$. Instead of expressing the constraint in terms of the transformed variable $Z_j$, we can express it in terms of the untransformed $X_j$:

\begin{equation}
\begin{split}
    \gamma_{0,\tau_{q}}&+\sum^K_{j=1} [\min(X_j) + (Z_j=0) \cdot (\max(X_j) - \min(X_j))]\gamma_{j,\tau_q}^+ \\
    &\geq \sum^K_{j=1} [\min(X_j) + (Z_j=1) \cdot (\max(X_j) - \min(X_j))]\gamma_{j,\tau_q}^- \nonumber
    \end{split}
\end{equation}

We can simplify this equation to recover:

\begin{equation}
    \gamma_{0,\tau_{q}}+\sum^K_{j=1} \min(X_j)\gamma_{j,\tau_q}^+ \geq \sum^K_{j=1} \max(X_j)\gamma_{j,\tau_q}^- \nonumber
\end{equation}

Thus there is an equivalence between this equation and the one in \citet{bondell2010noncrossing}.
\end{proof}

This new constraint is a sufficient condition for non-crossing, since if equation (\ref{eq:constraintEXPAND}) is satisfied for the worst case, it is immediately satisfied for every observation. The novelty of equation (\ref{eq:constraintEXPAND}) is that it works on the domain of $\mathbb{D} \in \mathbb{R}$, while the original formulation in \citet{bondell2010noncrossing}, works for $\mathbb{D} \in [0,1]$ only.

To study the impact of non-crossing constraints on the estimated coefficients, it is important to be able to tighten (and loosen) these constraints. As such, the next step is to formulate a set of constraints that can become adaptive as a hyperparameter $\alpha$ varies. An intuitively appealing formulation is one that yields the traditional quantile regression estimator when setting $\alpha=0$ and the \citet{bondell2010noncrossing} when $\alpha=1$.\footnote{Technically, any scalar $\alpha>0$ would work. We set $\alpha=1$ to recover the \citet{bondell2010noncrossing} for simplicity} This leads to the following adaptive non-crossing constraints:
\begin{equation}\label{eq:ADAconstraint}
        \gamma_{0,\tau_{q}}+\sum^K_{j=1} \Big[ \Bar{X_j} - \alpha(\Bar{X_j} - min(X_j)) \Big]\gamma_{j,\tau_q}^+\geq \sum^K_{j=1} \Big[\Bar{X_j} + \alpha(max(X_j)-\Bar{X_j}) \Big] \gamma_{j,\tau_q}^-
\end{equation}
When $\alpha=0$, the constraint simplifies to imposing non-crossing quantiles evaluated at the average value of the covariates. Quantile monotonicity at this value is satisfied even by the traditional quantile regression estimator: $\Bar{X}^T\beta$ yields the empirical quantiles of $Y$, which are monotonically increasing by definition \citep{koenker2005,koenker2006quantile}. As such imposing equation (\ref{eq:ADAconstraint}) and setting $\alpha=0$ will yield the same $\beta$ coefficients as the traditional quantile regression without constraints. 

The constraints in equation (\ref{eq:ADAconstraint}) equal the non-crossing constraints when $\alpha=1$ as the equation becomes equation (\ref{eq:constraintEXPAND}). As such using equation (\ref{eq:ADAconstraint}) as a constraint allows us to recover both the traditional quantile regression estimator of \citet{koenker1978regression}, as well as the non-crossing quantile regression estimator of \citet{bondell2010noncrossing}. 

There are other ways to construct adaptive non-crossing constraints that yield traditional quantile regression as well as \citet{bondell2010noncrossing} estimator, specifically through the use of indicator functions that activate the non-crossing constraint. The advantage of equation (\ref{eq:ADAconstraint}) is that it produces a gradual transition from qunatile estimates to non-crossing estimates as $\alpha$ increases. This allows us to study the impacts these constraints have on the estimated coefficients.

\begin{theorem} \label{theorem:NC=FLASSO}
Non-crossing constraints are a type of Fused LASSO, with quantile specific hyperparameters: $k_{\tau_q}^*= \frac{\gamma_{0,\tau_q}}{\alpha}$.

\end{theorem}

\begin{proof}
For simplicity, assume $Q=2$. We begin with the non-crossing constraint given by equation (\ref{eq:ADAconstraint}):
\begin{equation} \nonumber
\gamma_{0} \geq \sum^K_{j=1}\Big[[\Bar{X_j} + \alpha(max(X_j)-\Bar{X_j})]\gamma_{j}^- - [\Bar{X_j} - \alpha(\Bar{X_j} - min(X_j))]\gamma_{j}^+\Big]
\end{equation}

\noindent where $\gamma_{j}^-\geq0$ and $\gamma_{j}^+\geq0$ represent the positive and negative parts of $\gamma_{j}$ such that $\gamma_{j} = \gamma_{j}^+ - \gamma_{j}^-$.

\noindent By Lemma \ref{theorem:generalconstr}, we can rescale the data to $\mathbb{D} \in [-1,1]$. When data is properly standardized and rescaled, we can assume $\Bar{X_j} = 0$ for each covariate $j$. For a symmetric distribution, this standardization is trivial; for asymmetric distributions, one can normalise the data prior to rescaling. After normalising and rescaling, $\Bar{X_j} \approx 0$, $max(X_j) = 1$, and $min(X_j) = -1$. With this, the constraint simplifies to:
\begin{equation} \nonumber
\gamma_{0} \geq \sum^K_{j=1}\alpha( \gamma_{j}^- + \gamma_{j}^+)
\end{equation}

\noindent Since $\gamma_{j}^- \geq 0$ and $\gamma_{j}^+ \geq 0$ by definition, we have $(\gamma_{j}^- + \gamma_{j}^+) = |\gamma_{j}^- + \gamma_{j}^+|$. Rearranging terms yields:
\begin{equation}
\frac{\gamma_{0}}{\alpha} \geq \sum^K_{j=1} |\gamma^+_{j}+\gamma^-_{j}| \nonumber
\end{equation}

\noindent This corresponds exactly to the fused shrinkage formulation in \citet{jiang2013interquantile} with $k^* = \frac{\gamma_{0}}{\alpha}$.

For $Q>2$, all $\gamma$ parameters become quantile specific, leading to quantile specific hyperparameters: $k_{\tau_q}^* = \frac{\gamma_{0,\tau_q}}{\alpha}$. As such, non-crossing constraints lead to fused LASSO shrinkage of the parameters, with quantile specific hyperparameters. This equivalence requires $\alpha > 0$, as the relationship is undefined when $\alpha = 0$. As $\alpha\rightarrow0$, $k_{\tau_q}^*\rightarrow\infty$, and we recover the QR coefficients for sufficiently small $\alpha$.
\end{proof}

The results of theorem (\ref{theorem:NC=FLASSO}) mean that when $\alpha$ is large enough, the only way to satisfy the constraint is by having $\gamma_{j,\tau_q}^-=\gamma_{j,\tau_q}^+=0$, which is equivalent to the Composite Quantile Regression setup of \citet{koenker1984note} and \citet{zou2008composite}.\footnote{Note that in the limit there is no unique solution. Because of this, as $\alpha\rightarrow\infty$ it becomes numerically challenging to find the solution, leading to numerical instability.} Because it is possible to recover the traditional quantile regression estimator, the \citet{bondell2010noncrossing} estimator, and the composite quantile regression estimator by simply varying $\alpha$, we refer to the estimator using the constraints of equation (\ref{eq:ADAconstraint}) as Generalised Non-Crossing Quantile Regression (GNCQR).\footnote{The different types of estimators are defined, with a unified notation, in the appendix.}\textsuperscript{,}\footnote{We follow \citet{powell2020quantile} in the naming convention, who develop the `Generalized Quantile Regression'.} Formally GNCQR is defined as:
\begin{equation} \label{eq:GNCQR}
    \begin{split}
        \hat{\beta}_{GNCQR}|\alpha=\underset{\beta}{argmin}&\sum^{Q}_{q=1}\sum^{T}_{t=1}\rho_q(y_{t}-x_t^T\beta_{\tau_q})\\
        s.t.~\gamma_{0,\tau_{q}}+&\sum^K_{k=1} \Big[ \Bar{X_k} - \alpha(\Bar{X_k} - min(X_k)) \Big]\gamma_{k,\tau_q}^+\geq\\ &\sum^K_{k=1} \Big[\Bar{X_k} + \alpha(max(X_k)-\Bar{X_k}) \Big] \gamma_{k,\tau_q}^-
    \end{split}
\end{equation}

Equation (\ref{eq:GNCQR}) makes it clear that $\alpha$ plays a strong role in the estimator. Given that alpha plays a similar role to $\lambda$ in LASSO, we restrict its possible values to be $\alpha\geq0$. One might be tempted to restrict potential value of $\alpha\geq1$, i.e. always enforcing non-crossing quantiles in estimation. We caution against this, because the estimated conditional quantiles are merely an approximation of the true (and unknown) DGP \citep{koenker2006quantile}. In a situation where we have model misspecification, enforcing non-crossing constraints lead to too much interquantile shrinkage as a consequence of theorem (\ref{theorem:NC=FLASSO}).

The advantage of the GNCQR estimator is that it leads to quantile specific shrinkage parameters, while only needing to set the scalar $\alpha$. In essence, the problem becomes one of model selection, but instead of selecting variables from $X$, the interest is on interquantile shrinkage, i.e. selecting variables whose impact on $Y$ captured by $\beta$ varies by quantile.

\subsection{Hyperparameter selection}

\begin{figure}[t]
    \centering
\includegraphics[width=0.75\textwidth]{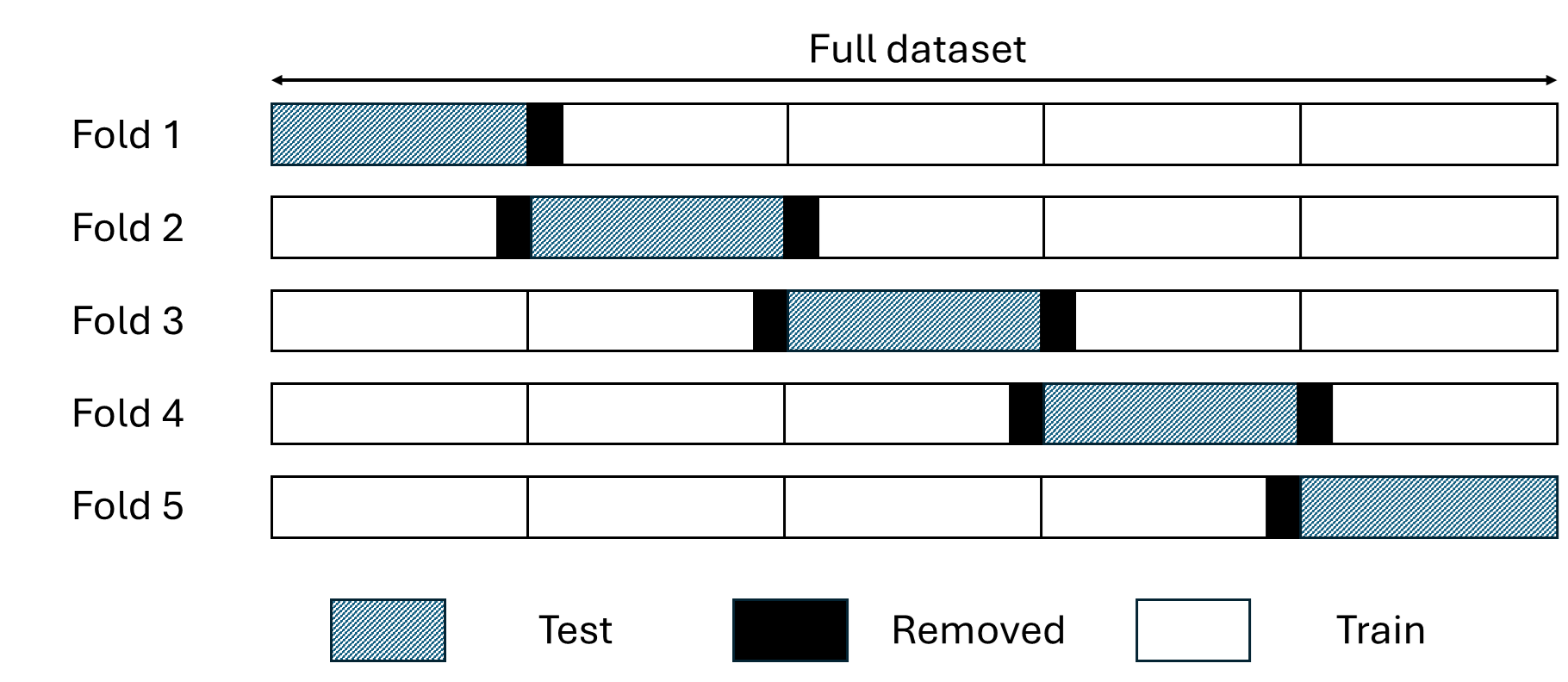}
\caption{hv-Cross Validation}
\label{fig:CV-fig}
\end{figure}

A natural candidate for choosing the degree of penalisation is cross-validation. In cross-section applications, this would be the usual $k$-fold cross-validation or a variant therefore. Of the several types of cross-validation methods available for time series application, we choose the $hv$-block CV setup of \citet{racine2000consistent}. This block setup has several desirable properties: (1)~it has been shown to have a good trade-off between bias and variance in various applications; (2)~the required number of computations does not increase with the number of observations to the degree it would with leave-one-out cross-validation,\footnote{See \citet{cerqueira2020evaluating} for a description and comparison of the different types of cross-validation for time-series data} and; (3)~when the ``$h$" in the $hv$-block is set equal to 0, we recover the standard $k$-fold cross-validation setup, so the method can be used for both cross-section and time-series data.\footnote{We use this equivalence in the Monte Carlo setups and applications so that the same cross-validation code can be employed throughout.}

Figure (\ref{fig:CV-fig}) visualises this block setup and which data are removed from the dataset. The reason to remove the data points around the test datasets in time-series data is to ensure that no data leakage occurs when evaluating the performance of the model. When data are dependent, the information from the test dataset can leak into the training set. This can inadvertently lead to overfitting and hence poor generalization. Removing the data around the test dataset can guard against this.\footnote{For more discussion on hyperparameter selection please see the appendix.}



Since the hyperparameter $\alpha$ is a scalar, one can find the optimal hyperparameter using simple grid-search. Hence, the curse of dimensionality that often limits the applicability of grid-search is not present here. Crucially, the grid-search is `embarrassingly parallel' and this can be utilised to speed up computation times \citep{bergstra2012random}. As such, throughout the paper we will utilise grid-search to obtain the optimal hyperparamter values.

\section{Monte Carlo}
\subsection{Setup}

Theorem \ref{theorem:NC=FLASSO} demonstrated the connection between Fused LASSO and non-crossing constraints. This theorem also shows that GNCQR is equivalent to having quantile-specific hyperparameters. In this section we explore the implications of this using Monte Carlo evidence, and  evaluate variable selection properties of the various estimators as well as examine their ability to recover the true quantiles. Specifically, we consider the following estimators: the non-crossing quantile regression of \citet{bondell2010noncrossing} (BRW), GNCQR, and the Fused LASSO (FLQR).\footnote{Note that we consider the Fused LASSO and not the Fused Adaptive LASSO in \citet{jiang2013interquantile}. We focus on the non-adaptive version as this allows us to examine the value added impact of quantile-specific hyperparameters.} We will also look at the two-step approach proposed in \citet{chernozhukov2010crossing}, i.e. running regular quantile regression and sorting the fitted quantiles. This will allow us to gauge how much better model selection results lead to improved fit.

Our Monte Carlo setup takes as its starting point the design used in \citet{bondell2010noncrossing}. Each Monte Carlo experiment is generated using the location scale heteroskedastic error model of the form:
\begin{equation} \label{eq:loc-scale}
    y_t=\beta_0+\beta^Tx_t+(\theta_0+\vartheta_t  \odot \theta^Tx_t)\varepsilon_t,~x_{t,k}\sim U(0,1),~\varepsilon_t\sim N(0,1)
\end{equation}

An intercept is included in each setup (i.e., $\beta_0=\theta_0=1$). Note the $k^{th}$ element of $\vartheta_{t}\in \{0,1\}$ which regulates whether the given variable has quantile variation at the given quantile. This term is included to allow `quantile varying sparsity', i.e. cases where a certain variable enters only parts of the distribution \citep{kohns2021decoupling}. A simple way to implement quantile-specific sparsity is by setting $\vartheta_t$ as an indicator function, where it takes the value of 1 only for cases when $\varepsilon_t$ is below (or above) a specific quantile. For cases where there is full quantile variation, $\vartheta_t=\textbf{1}_k$. The $t$ subscript is needed the presence of quantile variation will be dependent on the magnitude of $\varepsilon_t$.

There are four data generating processes (DGPs), each with 500 generated datasets. The first three ($y_1$, $y_2$, and $y_3$) are identical to Examples 1-3 of \citet{bondell2010noncrossing}. The fourth DGP ($y_4$) is an amendment of Example 2, with some variables only varying at the tails. Specifically:
\begin{itemize}
    \item $y_1$: $k=4$, with the parameters $\beta=\textbf{1}_k$, $\theta=0.1*\textbf{1}_k$, and $\vartheta=\textbf{1}_k$.
    \item $y_2$: $k=10$, with the parameters $\beta=(\textbf{1}_4^T,\textbf{0}_6^T)^T$, $\theta=(0.1*\textbf{1}_4^T,\textbf{0}_6^T)^T$, and $\vartheta=\textbf{1}_k$.
    \item $y_3$: $k=7$, with the parameters $\beta=\textbf{1}_k$, $\theta=(\textbf{1}_3^T,\textbf{0}_4^T)^T$, and $\vartheta=\textbf{1}_k$.
    \item $y_4$: $k=10$, with the parameters $\beta=(\textbf{1}_4^T,\textbf{0}_6^T)^T$, $\theta=(0.1*\textbf{1}_8^T,\textbf{0}_2^T)^T$, and $\vartheta=(\textbf{1}_4^T,\textbf{1}_4^T \times [I(\epsilon_t> F^{-1}_\varepsilon(0.9))+I(\epsilon_t\leq F^{-1}_\varepsilon(0.1))],\textbf{0}_2^T)^T$.
\end{itemize}

Note that for $y_4$, variable selection and fused shrinkage will both be necessary. The models considered will only allow for fused shrinkage, and as such this DGP is only included to judge the performance of the estimators in unideal situations.

To test the performance of the different estimators three sample sizes are considered for each DGP: $T\in\{50,100,200\}$. Sample sizes 100 and 200 were also considered in \citet{bondell2010noncrossing}, but 50 was not. We include this small sample setting because in macroeconometric applications it is not uncommon to apply quantile regression for such small samples.\footnote{See \citet{szendrei2023revisiting} for an application with around 50 observations, \citet{figueres2020vulnerable} for an application with around 100 observations, and \citet{adrian2019vulnerable} for an application with around 200 observations.} Applying quantile regression to analyse GDP distribution for European economies has been particularly daunting on account of data availability, so we feel that these Monte Carlo results are particularly useful for policy makers.

We also consider variation in the number of quantiles to be estimated, by  generating equidistant quantiles with varying distance between the quantiles: $\Delta_\tau\in\{0.2,0.1\}$.\footnote{for$\Delta_\tau=0.2$, the first quantile is set to 0.1.} Note that by increasing the number of quantiles, the number of estimated parameters varies. The choice to vary the estimated quantiles was driven by the fact that Theorem (\ref{theorem:NC=FLASSO}) shows how $\gamma_0$ acts as a quantile-specific hyperparameter, which is a claim we can verify by examining the variable selection properties of BRW as the number of quantiles increases.

Two measures are used to judge the variable selection performance of the estimators: True Positive Rate (TPR) and True Negative Rate (TNR). Both measures take a value between 0 (worst performance possible) and 1 (best performance possible). TPR measures the degree to which the estimator is able to capture the quantile varying coefficients, while the TNR measures the ability of the estimator to identify where no quantile variation occurs. Considering these measures together allows one to conclude whether a given estimator over-shrinks or under-shrinks. When calculating these measures, we only consider the difference in $\beta$ coefficients of the variables (without the intercept).

To measure fit, we follow \citet{bondell2010noncrossing} and report the average root mean integrated square error ($\times 100$):

\begin{equation}
    RMISE=\frac{1}{n}\Big[ \sum^{n}_{iter=1}\Big(\hat{g}_\tau(x_{iter})-g_\tau(x_{iter}) \Big)^2 \Big]^{1/2} \nonumber
\end{equation}
where $iter$ indexes a given Monte Carlo experiment, $n=500$ is the total number of evaluated Monte Carlo experiments for each DGP, $\hat{g}_\tau$  is the estimated quantile and $g_\tau$ is the true quantile given equation (\ref{eq:loc-scale}). We also report the standard error of the RMISE for all the estimators.

We create a grid of 100 equally-spaced points between $[0,1)$, and use grid-search to obtain the optimal $\alpha$ parameter. To consider large $\alpha$ options, we append to this list of candidate solutions 200 points between $[0,6]$ as exponents of base 10. Setting the hyperparameter to $10^6$ yields solutions that are close to the composite quantile regression solution. For FLQR and GNCQR, we use 10-fold cross-validation.

Note that while Theorem (\ref{theorem:NC=FLASSO}) implies that changing the $\alpha$ parameter will have a similar impact on the $\beta$'s as the $\lambda$ of a Fused LASSO estimator, the same $\alpha$ and $\lambda$ values will not lead to the same $\beta$ coefficients. This is because the quantile-specific difference in constants is an upper limit of quantile variation for GNCQR.


Table (\ref{tab:MC_RMISE1}) presents the results for the average fit of the different estimators based on the Monte Carlo experiments, while Table \ref{tab:MC_TXR} shows the variable selection performance. The RIMSE set of results for $y_1$, $y_2$, and $y_3$ when T=100 (and T=200 for $y_3$) and $\Delta\tau=0.2$ is the closest setup to \citet{bondell2010noncrossing}. The results for these setups for BRW and QR are almost identical and as such our new findings relate to the extensions to the Monte Carlo experiments of \citet{bondell2010noncrossing}.

\subsection{Variable Selection results}

\begin{table}[]
\centering
\caption{True Positive and True Negative Rates for the different Monte Carlo experiments}
\label{tab:MC_TXR}
\resizebox{1\columnwidth}{!}{%
\begin{tabular}{lr|ccccccc|ccccccc}
\hline
 &  & \multicolumn{1}{c}{$y_1$} & \multicolumn{2}{c}{$y_2$} & \multicolumn{2}{c}{$y_3$} & \multicolumn{2}{c|}{$y_4$} & \multicolumn{1}{c}{$y_1$} & \multicolumn{2}{c}{$y_2$} & \multicolumn{2}{c}{$y_3$} & \multicolumn{2}{c}{$y_4$} \\
 &  & TPR & TPR & TNR & TPR & TNR & TPR & TNR & TPR & TPR & TNR & TPR & TNR & TPR & TNR \\ \hline
 &  & \multicolumn{7}{|c|}{$\Delta\tau=0.2$} & \multicolumn{7}{c}{$\Delta\tau=0.1$}  \\
\multicolumn{2}{l|}{T-50} &  &  &  &  &  &  &  &  \\
 & BRW & 0.806 & 0.360 & 0.639 & 0.535 & 0.525 & 0.305 & 0.691 & 0.618 & 0.246 & 0.752 & 0.383 & 0.669 & 0.224 & 0.776 \\
 & GNCQR & 0.372 & 0.232 & 0.769 & 0.345 & 0.715 & 0.243 & 0.753 & 0.304 & 0.195 & 0.812 & 0.271 & 0.788 & 0.173 & 0.825 \\
 & QR & 1.000 & 1.000 & 0.000 & 1.000 & 0.000 & 1.000 & 0.000 & 1.000 & 1.000 & 0.000 & 1.000 & 0.000 & 1.000 & 0.000 \\
 & FLQR & 0.215 & 0.136 & 0.868 & 0.246 & 0.790 & 0.194 & 0.819 & 0.127 & 0.107 & 0.901 & 0.154 & 0.870 & 0.114 & 0.888 \\
\multicolumn{2}{l|}{T=100} &  &  &  &  &  &  &  &  \\
 & BRW & 0.938 & 0.510 & 0.501 & 0.694 & 0.393 & 0.390 & 0.582 & 0.792 & 0.363 & 0.641 & 0.531 & 0.557 & 0.318 & 0.687 \\
 & GNCQR & 0.364 & 0.228 & 0.780 & 0.430 & 0.689 & 0.279 & 0.715 & 0.294 & 0.183 & 0.828 & 0.349 & 0.767 & 0.179 & 0.833 \\
 & QR & 1.000 & 1.000 & 0.000 & 1.000 & 0.000 & 1.000 & 0.000 & 1.000 & 1.000 & 0.000 & 1.000 & 0.000 & 1.000 & 0.000 \\
 & FLQR & 0.202 & 0.122 & 0.885 & 0.299 & 0.776 & 0.184 & 0.834 & 0.110 & 0.074 & 0.927 & 0.196 & 0.860 & 0.078 & 0.932 \\
\multicolumn{2}{l|}{T=200} &  &  &  &  &  &  &  &  \\
 & BRW & 0.988 & 0.659 & 0.340 & 0.846 & 0.262 & 0.481 & 0.459 & 0.923 & 0.488 & 0.517 & 0.690 & 0.445 & 0.409 & 0.575 \\
 & GNCQR & 0.381 & 0.225 & 0.787 & 0.564 & 0.672 & 0.280 & 0.700 & 0.325 & 0.166 & 0.846 & 0.465 & 0.765 & 0.203 & 0.797 \\
 & QR & 1.000 & 1.000 & 0.000 & 1.000 & 0.000 & 1.000 & 0.000 & 1.000 & 1.000 & 0.000 & 1.000 & 0.000 & 1.000 & 0.000 \\
 & FLQR & 0.215 & 0.130 & 0.883 & 0.413 & 0.746 & 0.209 & 0.821 & 0.139 & 0.061 & 0.946 & 0.298 & 0.832 & 0.098 & 0.917 \\ \hline
\end{tabular}%
}
\end{table}

The results on variable selection presented in Table \ref{tab:MC_TXR} are more revealing. Since the GNCQR can recover both the BRW (when setting $\alpha=1$) and the QR (as $\alpha\rightarrow0$), we can compare the performance of these estimators to see the impact $\alpha$ has on the coefficient profiles. Given the findings of Theorem \ref{theorem:NC=FLASSO}, we expect the TPR to decrease and the TNR to increase as $\alpha$ increases. We find that the TPR of BRW is below 1 (and the TNR is above 0) for all DGPs, which corroborates Theorem \ref{theorem:NC=FLASSO}, which states that non-crossing constraints are a special type of fused shrinkage. This is not simply a feature of the Monte Carlo design, as the QR yields a TPR of 1 and a TNR of 0 in all cases. Note that for $y_1$, all variables are quantile-varying and as such TNR does not exist for this case.

Among the estimators that have some fused shrinkage, BRW yields the highest TPR for all DGPs for all cases considered. GNCQR always ranks second and FLQR has the worst performance for TPR. However, this superior performance in TPR for BRW is coupled with the worst performance when it comes to TNR. For TNR, FLQR produces the best results, with GNCQR a close second. From these results we can see that BRW undershrinks quantile variation, FLQR overshrinks, and GNCQR yields a middle-ground option.

As will be seen in the fit results of Table (\ref{tab:MC_RMISE1}), we find that GNCQR not only provides robustness over FLQR for $y_4$, but is also able to give fits close to FLQR while retaining better model selection properties. In particular, GNCQR is able to yield much better TPR than FLQR without substantially compromising its ability to identify the true negative differences.

Increasing the number of quantiles has a marked impact on variable selection: it lowers TPR and increases TNR for all estimators (except the QR). This further corroborates Theorem \ref{theorem:NC=FLASSO}. Increasing sample size also influences the TPR and TNR of all estimators. For BRW and GNCQR, increasing the sample size increases TPR but lowers TNR for all DGPs and all $\Delta_\tau$. However, for FLQR we see similar tendency when $\Delta_\tau=0.2$ but not when $\Delta_\tau=0.1$. In particular, for FLQR for $y_4$, the TPR decreases as the sample size increases when $\Delta_\tau=0.1$. This is particularly troubling since FLQR has the worst TPR of all the estimators.

\subsection{Fit results}
Results on fit in Table (\ref{tab:MC_RMISE1}) reveal that GNCQR and FLQR provide best fits, even beating BRW estimator, for $y_1$, $y_2$, and $y_3$. However, for $y_4$ FLQR fails to yield improvements over BRW. While GNCQR also has difficulties in $y_4$, it is much closer to the fits of BRW, with both yielding better fits than the traditional QR or FLQR. This is because GNCQR recovers BRW when $\alpha=1$, so it will not do much worse than BRW. Hence, GNCQR is more robust to DGPs that have quantile-specific sparsity than the simple FLQR.

Unsurprisingly, $y_4$ produces the worst fit for all estimators, but as more data becomes available, the performance of all the estimators initially improves. The key difference lies with FLQR, where the fits do not improve as much as the other estimators when the sample size increases from $T=100$ to $T=200$. This also highlights that to yield improvements in the traditional LASSO setting for such DGPs one needs to explicitly penalise the level of $\beta_\tau$ coefficients too; see also \citet{jiang2014interquantile}. Overall, increasing sample size improves the fit for all estimators, while increasing the fineness of the grid of quantiles yields no significant differences. The key takeaway is that GNCQR either gives further improvements in fit over BRW, or (at worst) does as well as BRW.

Looking at the rearrangement method proposed by \citet{chernozhukov2010crossing} reveals that sorting the fitted quantiles for the quantile regression estimator yields modest improvements. These improvements increase as the number of estimated quantiles increases especially for smaller sample sizes. This is intuitive, as with more quantiles there is a higher chance of quantile crossing. However, the improvements in fit from rearrangement get smaller as the sample size increases.

\begingroup
\setlength{\tabcolsep}{2pt} 
\begin{landscape}
\begin{table}[]
\centering
\caption{RIMSE of different models for the different Monte Carlo experiments}
\label{tab:MC_RMISE1}
\resizebox{\columnwidth}{!}{%
\begin{tabular}{lr|cccccccccc|cccccccccc|cccccccccc}
\hline
 &  & \multicolumn{10}{c|}{T=50} & \multicolumn{10}{c|}{T=100} & \multicolumn{10}{c}{T=200} \\
 & $\tau$ & \multicolumn{2}{c}{0.1} & \multicolumn{2}{c}{0.3} & \multicolumn{2}{c}{0.5} & \multicolumn{2}{c}{0.7} & \multicolumn{2}{c|}{0.9} & \multicolumn{2}{c}{0.1} & \multicolumn{2}{c}{0.3} & \multicolumn{2}{c}{0.5} & \multicolumn{2}{c}{0.7} & \multicolumn{2}{c|}{0.9} & \multicolumn{2}{c}{0.1} & \multicolumn{2}{c}{0.3} & \multicolumn{2}{c}{0.5} & \multicolumn{2}{c}{0.7} & \multicolumn{2}{c}{0.9} \\
 & \multicolumn{1}{l}{} & \multicolumn{1}{|l}{bias} & \multicolumn{1}{l}{std. err} & \multicolumn{1}{l}{bias} & \multicolumn{1}{l}{std. err} & \multicolumn{1}{l}{bias} & \multicolumn{1}{l}{std. err} & \multicolumn{1}{l}{bias} & \multicolumn{1}{l}{std. err} & \multicolumn{1}{l}{bias} & \multicolumn{1}{l|}{std. err} & \multicolumn{1}{l}{bias} & \multicolumn{1}{l}{std. err} & \multicolumn{1}{l}{bias} & \multicolumn{1}{l}{std. err} & \multicolumn{1}{l}{bias} & \multicolumn{1}{l}{std. err} & \multicolumn{1}{l}{bias} & \multicolumn{1}{l}{std. err} & \multicolumn{1}{l}{bias} & \multicolumn{1}{l|}{std. err} & \multicolumn{1}{l}{bias} & \multicolumn{1}{l}{std. err} & \multicolumn{1}{l}{bias} & \multicolumn{1}{l}{std. err} & \multicolumn{1}{l}{bias} & \multicolumn{1}{l}{std. err} & \multicolumn{1}{l}{bias} & \multicolumn{1}{l}{std. err} & \multicolumn{1}{l}{bias} & \multicolumn{1}{l}{std. err} \\ \hline \hline
\multicolumn{2}{l|}{$\Delta\tau=0.2$} &  &  &  &  &  &  &  &  &  &  &  &  &  &  &  &  &  &  &  &  &  &  &  &  &  &  &  &  &  &  \\
\multicolumn{2}{l|}{$y_1$} &  &  &  &  &  &  &  &  &  &  &  &  &  &  &  &  &  &  &  &  &  &  &  &  &  &  &  &  &  &  \\
 & BRW & 55.3 & 0.76 & 44.3 & 0.65 & 42.8 & 0.62 & 44.8 & 0.66 & 56.5 & 0.80 & 41.7 & 0.58 & 32.9 & 0.47 & 30.3 & 0.45 & 32.0 & 0.49 & 41.1 & 0.61 & 30.4 & 0.45 & 23.5 & 0.33 & 22.1 & 0.32 & 23.2 & 0.33 & 30.4 & 0.42 \\
 & GNCQR & 46.5 & 0.72 & 41.7 & 0.61 & 40.8 & 0.61 & 41.9 & 0.63 & 48.2 & 0.78 & 33.6 & 0.54 & 29.2 & 0.44 & 27.8 & 0.42 & 29.3 & 0.45 & 33.6 & 0.57 & 23.9 & 0.38 & 20.9 & 0.30 & 20.0 & 0.29 & 20.5 & 0.30 & 24.3 & 0.39 \\
 & QR & 60.7 & 0.84 & 46.8 & 0.68 & 45.3 & 0.65 & 48.0 & 0.71 & 62.4 & 0.93 & 43.9 & 0.63 & 33.9 & 0.50 & 31.2 & 0.47 & 33.0 & 0.49 & 43.5 & 0.66 & 30.8 & 0.46 & 23.8 & 0.33 & 22.4 & 0.33 & 23.5 & 0.33 & 31.0 & 0.44 \\
 & FLQR & 46.6 & 0.72 & 41.8 & 0.61 & 40.6 & 0.61 & 42.1 & 0.65 & 48.2 & 0.77 & 33.2 & 0.55 & 29.5 & 0.45 & 27.9 & 0.44 & 29.3 & 0.45 & 34.2 & 0.59 & 24.1 & 0.39 & 21.0 & 0.29 & 20.1 & 0.30 & 20.6 & 0.30 & 24.9 & 0.41 \\
\multicolumn{2}{l|}{$y_2$} &  &  &  &  &  &  &  &  &  &  &  &  &  &  &  &  &  &  &  &  &  &  &  &  &  &  &  &  &  &  \\
 & BRW & 73.4 & 0.69 & 63.0 & 0.58 & 61.2 & 0.57 & 63.1 & 0.62 & 73.2 & 0.72 & 55.0 & 0.51 & 45.9 & 0.44 & 44.6 & 0.43 & 46.4 & 0.42 & 54.4 & 0.52 & 40.8 & 0.39 & 33.1 & 0.32 & 31.7 & 0.30 & 33.2 & 0.32 & 40.7 & 0.37 \\
 & GNCQR & 69.4 & 0.75 & 61.9 & 0.60 & 60.1 & 0.58 & 62.0 & 0.63 & 68.3 & 0.76 & 48.7 & 0.53 & 43.8 & 0.42 & 43.3 & 0.42 & 44.0 & 0.43 & 48.5 & 0.53 & 34.3 & 0.38 & 30.6 & 0.31 & 29.9 & 0.29 & 30.8 & 0.30 & 33.8 & 0.37 \\
 & QR & 94.5 & 0.84 & 72.2 & 0.70 & 67.6 & 0.67 & 72.3 & 0.73 & 94.8 & 0.90 & 67.1 & 0.65 & 51.9 & 0.47 & 49.3 & 0.47 & 52.2 & 0.49 & 67.9 & 0.66 & 47.6 & 0.47 & 36.5 & 0.35 & 34.6 & 0.33 & 36.7 & 0.36 & 47.2 & 0.45 \\
 & FLQR & 68.5 & 0.74 & 61.8 & 0.61 & 59.9 & 0.59 & 61.7 & 0.62 & 69.3 & 0.78 & 47.9 & 0.54 & 43.6 & 0.43 & 43.3 & 0.42 & 44.0 & 0.43 & 49.5 & 0.55 & 33.8 & 0.38 & 30.6 & 0.32 & 30.1 & 0.30 & 30.9 & 0.31 & 34.7 & 0.41 \\
\multicolumn{2}{l|}{$y_3$} &  &  &  &  &  &  &  &  &  &  &  &  &  &  &  &  &  &  &  &  &  &  &  &  &  &  &  &  &  &  \\
 & BRW & 136.8 & 1.62 & 110.4 & 1.28 & 108.3 & 1.28 & 114.2 & 1.36 & 137.6 & 1.70 & 101.3 & 1.19 & 79.7 & 0.94 & 76.2 & 0.83 & 80.8 & 0.91 & 100.0 & 1.19 & 75.7 & 0.83 & 58.8 & 0.65 & 54.8 & 0.62 & 57.6 & 0.65 & 74.0 & 0.83 \\
 & GNCQR & 132.5 & 1.71 & 109.4 & 1.30 & 106.9 & 1.27 & 112.5 & 1.32 & 133.9 & 1.69 & 98.6 & 1.22 & 78.0 & 0.92 & 74.7 & 0.83 & 80.0 & 0.92 & 98.6 & 1.16 & 75.2 & 0.83 & 57.2 & 0.63 & 52.9 & 0.62 & 56.9 & 0.65 & 74.4 & 0.80 \\
 & QR & 167.7 & 2.00 & 124.1 & 1.38 & 118.2 & 1.39 & 127.8 & 1.49 & 165.7 & 1.95 & 116.8 & 1.37 & 87.8 & 1.01 & 82.7 & 0.91 & 89.3 & 0.97 & 115.2 & 1.37 & 83.1 & 0.96 & 62.8 & 0.71 & 58.1 & 0.65 & 61.8 & 0.70 & 81.1 & 0.91 \\
 & FLQR & 133.7 & 1.74 & 109.7 & 1.30 & 106.9 & 1.29 & 113.5 & 1.37 & 136.4 & 1.67 & 99.9 & 1.18 & 78.3 & 0.92 & 75.8 & 0.86 & 81.7 & 0.92 & 100.1 & 1.13 & 77.9 & 0.85 & 57.4 & 0.64 & 53.8 & 0.62 & 58.6 & 0.66 & 74.5 & 0.86 \\
\multicolumn{2}{l|}{$y_4$} &  &  &  &  &  &  &  &  &  &  &  &  &  &  &  &  &  &  &  &  &  &  &  &  &  &  &  &  &  &  \\
 & BRW & 225.1 & 2.40 & 85.9 & 1.61 & 77.3 & 1.27 & 85.5 & 1.66 & 224.8 & 2.29 & 194.4 & 2.28 & 53.1 & 0.66 & 48.4 & 0.55 & 52.9 & 0.63 & 197.5 & 2.22 & 176.9 & 2.32 & 37.4 & 0.40 & 33.6 & 0.36 & 37.4 & 0.38 & 180.5 & 2.34 \\
 & GNCQR & 224.7 & 3.06 & 85.5 & 1.66 & 77.8 & 1.36 & 85.4 & 1.78 & 222.9 & 2.97 & 196.5 & 3.18 & 52.2 & 0.68 & 48.6 & 0.59 & 51.9 & 0.66 & 199.8 & 3.06 & 181.6 & 3.28 & 36.1 & 0.40 & 34.1 & 0.37 & 36.0 & 0.39 & 186.0 & 3.33 \\
 & QR & 279.7 & 1.74 & 100.0 & 2.34 & 72.5 & 1.09 & 98.9 & 2.32 & 282.4 & 1.78 & 231.0 & 1.42 & 54.5 & 0.81 & 48.7 & 0.51 & 54.7 & 0.75 & 232.2 & 1.44 & 199.5 & 1.50 & 36.7 & 0.41 & 34.2 & 0.35 & 36.8 & 0.37 & 201.6 & 1.52 \\
 & FLQR & 228.3 & 3.40 & 85.7 & 1.66 & 79.0 & 1.38 & 85.8 & 1.84 & 228.6 & 2.91 & 200.9 & 3.55 & 51.9 & 0.67 & 49.2 & 0.61 & 51.2 & 0.65 & 205.7 & 3.06 & 195.3 & 3.39 & 35.4 & 0.40 & 34.0 & 0.37 & 35.1 & 0.38 & 192.2 & 3.11 \\ \hline
 \multicolumn{2}{l|}{$\Delta\tau=0.1$} &  &  &  &  &  &  &  &  &  &  &  &  &  &  &  &  &  &  &  &  &  &  &  &  &  &  &  &  &  &  \\
\multicolumn{2}{l|}{$y_1$} &  &  &  &  &  &  &  &  &  &  &  &  &  &  &  &  &  &  &  &  &  &  &  &  &  &  &  &  &  &  \\
 & BRW & 54.3 & 0.76 & 44.0 & 0.64 & 42.8 & 0.62 & 44.8 & 0.66 & 55.6 & 0.78 & 40.6 & 0.57 & 32.5 & 0.46 & 29.8 & 0.44 & 31.6 & 0.49 & 40.0 & 0.59 & 29.7 & 0.44 & 23.3 & 0.33 & 21.8 & 0.33 & 22.9 & 0.32 & 29.7 & 0.42 \\
 & GNCQR & 46.2 & 0.72 & 41.6 & 0.61 & 40.7 & 0.61 & 41.7 & 0.64 & 47.6 & 0.76 & 33.5 & 0.53 & 29.2 & 0.45 & 27.7 & 0.42 & 28.9 & 0.44 & 32.7 & 0.54 & 24.0 & 0.37 & 20.9 & 0.29 & 19.9 & 0.29 & 20.3 & 0.29 & 23.7 & 0.37 \\
 & QR & 60.7 & 0.84 & 46.8 & 0.68 & 45.3 & 0.65 & 48.0 & 0.71 & 62.4 & 0.93 & 43.9 & 0.63 & 33.9 & 0.50 & 31.2 & 0.47 & 33.0 & 0.49 & 43.5 & 0.66 & 30.8 & 0.46 & 23.8 & 0.33 & 22.4 & 0.33 & 23.5 & 0.33 & 31.0 & 0.44 \\
 & FLQR & 46.4 & 0.72 & 41.5 & 0.61 & 40.7 & 0.60 & 42.0 & 0.64 & 48.1 & 0.81 & 32.5 & 0.53 & 29.0 & 0.44 & 27.8 & 0.42 & 29.1 & 0.44 & 33.2 & 0.56 & 24.2 & 0.39 & 20.9 & 0.29 & 20.0 & 0.29 & 20.3 & 0.30 & 24.0 & 0.38 \\
\multicolumn{2}{l|}{$y_2$} &  &  &  &  &  &  &  &  &  &  &  &  &  &  &  &  &  &  &  &  &  &  &  &  &  &  &  &  &  &  \\
 & BRW & 73.2 & 0.69 & 63.1 & 0.59 & 61.4 & 0.58 & 63.2 & 0.61 & 72.3 & 0.71 & 54.3 & 0.50 & 45.9 & 0.43 & 44.6 & 0.43 & 46.1 & 0.43 & 53.6 & 0.52 & 40.0 & 0.38 & 32.9 & 0.32 & 31.6 & 0.30 & 33.0 & 0.32 & 39.9 & 0.37 \\
 & GNCQR & 69.9 & 0.74 & 62.4 & 0.61 & 60.6 & 0.59 & 62.5 & 0.63 & 69.1 & 0.78 & 48.8 & 0.53 & 43.8 & 0.42 & 43.2 & 0.43 & 43.9 & 0.43 & 48.6 & 0.54 & 33.6 & 0.36 & 30.3 & 0.30 & 29.7 & 0.29 & 30.5 & 0.30 & 33.3 & 0.35 \\
 & QR & 94.5 & 0.84 & 72.2 & 0.70 & 67.6 & 0.67 & 72.3 & 0.73 & 94.8 & 0.90 & 67.1 & 0.65 & 51.9 & 0.47 & 49.3 & 0.47 & 52.2 & 0.49 & 67.9 & 0.66 & 47.6 & 0.47 & 36.5 & 0.35 & 34.6 & 0.33 & 36.7 & 0.36 & 47.2 & 0.45 \\
 & FLQR & 69.6 & 0.76 & 62.3 & 0.63 & 60.3 & 0.60 & 62.4 & 0.63 & 69.6 & 0.79 & 48.0 & 0.53 & 43.9 & 0.43 & 43.1 & 0.42 & 43.9 & 0.43 & 48.8 & 0.55 & 33.1 & 0.35 & 30.3 & 0.31 & 29.6 & 0.29 & 30.4 & 0.30 & 33.3 & 0.38 \\
\multicolumn{2}{l|}{$y_3$} &  &  &  &  &  &  &  &  &  &  &  &  &  &  &  &  &  &  &  &  &  &  &  &  &  &  &  &  &  &  \\
 & BRW & 135.2 & 1.62 & 110.6 & 1.26 & 108.0 & 1.25 & 113.6 & 1.36 & 137.3 & 1.67 & 100.1 & 1.18 & 79.6 & 0.91 & 76.2 & 0.82 & 80.5 & 0.89 & 98.0 & 1.19 & 73.9 & 0.82 & 57.9 & 0.64 & 54.1 & 0.62 & 57.3 & 0.67 & 73.0 & 0.81 \\
 & GNCQR & 130.9 & 1.67 & 109.7 & 1.25 & 106.5 & 1.27 & 112.1 & 1.32 & 133.3 & 1.65 & 97.5 & 1.20 & 77.8 & 0.90 & 74.1 & 0.82 & 78.2 & 0.90 & 97.3 & 1.14 & 74.4 & 0.80 & 56.8 & 0.63 & 52.7 & 0.61 & 56.2 & 0.65 & 73.1 & 0.80 \\
 & QR & 167.7 & 2.00 & 124.1 & 1.38 & 118.2 & 1.39 & 127.8 & 1.49 & 165.7 & 1.95 & 116.8 & 1.37 & 87.8 & 1.01 & 82.7 & 0.91 & 89.3 & 0.97 & 115.2 & 1.37 & 83.1 & 0.96 & 62.8 & 0.71 & 58.1 & 0.65 & 61.8 & 0.70 & 81.1 & 0.91 \\
 & FLQR & 132.6 & 1.73 & 110.9 & 1.29 & 106.7 & 1.30 & 113.5 & 1.34 & 135.2 & 1.64 & 100.6 & 1.19 & 78.6 & 0.96 & 74.6 & 0.84 & 79.4 & 0.90 & 98.0 & 1.14 & 78.2 & 0.83 & 57.4 & 0.64 & 53.4 & 0.62 & 57.8 & 0.65 & 74.4 & 0.85 \\
\multicolumn{2}{l|}{$y_4$} &  &  &  &  &  &  &  &  &  &  &  &  &  &  &  &  &  &  &  &  &  &  &  &  &  &  &  &  &  &  \\
 & BRW & 224.3 & 2.52 & 84.2 & 1.63 & 75.5 & 1.26 & 84.1 & 1.67 & 223.2 & 2.36 & 193.5 & 2.37 & 50.9 & 0.64 & 47.0 & 0.53 & 51.1 & 0.64 & 196.9 & 2.30 & 175.4 & 2.43 & 35.6 & 0.41 & 32.5 & 0.34 & 35.4 & 0.36 & 179.4 & 2.45 \\
 & GNCQR & 222.0 & 3.33 & 82.6 & 1.64 & 75.4 & 1.29 & 82.9 & 1.78 & 217.1 & 3.22 & 190.6 & 3.53 & 49.1 & 0.63 & 47.1 & 0.56 & 49.0 & 0.62 & 194.5 & 3.33 & 179.8 & 3.51 & 33.7 & 0.39 & 32.3 & 0.35 & 33.4 & 0.37 & 182.9 & 3.54 \\
 & QR & 279.7 & 1.74 & 100.0 & 2.34 & 72.5 & 1.09 & 98.9 & 2.32 & 282.4 & 1.78 & 231.0 & 1.42 & 54.5 & 0.81 & 48.7 & 0.51 & 54.7 & 0.75 & 232.2 & 1.44 & 199.5 & 1.50 & 36.7 & 0.41 & 34.2 & 0.35 & 36.8 & 0.37 & 201.6 & 1.52 \\
 & FLQR & 226.2 & 3.60 & 84.5 & 1.84 & 75.5 & 1.30 & 82.7 & 1.80 & 222.7 & 3.27 & 194.0 & 3.75 & 48.7 & 0.60 & 47.3 & 0.57 & 48.6 & 0.62 & 199.9 & 3.36 & 190.5 & 3.64 & 33.1 & 0.37 & 32.1 & 0.35 & 33.0 & 0.36 & 190.0 & 3.41 \\ \hline
\end{tabular}%
}
\end{table}
\end{landscape}
\endgroup

Comparing the rearrangement method with the alternative estimators shows that although rearrangement helps improve fit, the alternative estimators still perform better than the sorted QR. This highlights how better interquantile variable selection translates to superior fit beyond simple quantile sorting.

Note that in all the Monte Carlo runs, we have estimated a correctly specified model. In such situations the BRW estimator will (almost) always yield improvements over the QR. When the estimated model is misspecified, imposing strict non-crossing constraints can lead to worse coefficient bias than the QR. The fact that the GNCQR allows for some quantile crossing when $\alpha<1$, means that the estimator is more robust to misspecification than the BRW. In essence, when $0<\alpha_{opt}<1$ the GNCQR provides the best linear approximation of the quantiles while regularising some of the quantile variation in the coefficients.

In summary, GNCQR and FLQR provides better fit than BRW for $y_1$, $y_2$, and $y_3$, while for $y_4$ BRW performs best but GNCQR is nearly as good and both are preferred over the FLQR. This robustness of GNCQR is attractive in quantile applications as it is difficult to know ex ante whether quantile specific sparsity is present in the DGP. The variable selection results demonstrate that BRW undershrinks while FLQR overshrinks. Taking the fit and variable selection results together, we can see that GNCQR provides good fit while retaining better variable selection properties than FLQR.

\section{Growth-at-Risk results}
The canonical GaR uses US quarterly GDP growth in conjunction with the NFCI to obtain estimates of downside risk. Using quantile regression, GaR estimates can be obtained by estimating the following model:
\begin{equation}
    y_{t+h} = x_t'\beta_{\tau_q} + \varepsilon_{t+h}
\end{equation}

for $t=1,\cdots,T-h$, where $h$ refers to the forecast horizon. $x_t$ includes a constant (intercept), current quarterly GDP growth (annualised), and the average of the NFCI for the given quarter. The quarterly data cover 1973Q1 to 2023Q1. This sample includes the original sample of \citet{adrian2019vulnerable}, augmented by the COVID-19 crisis. We consider one- and four-quarter ahead forecast horizons ($h=1,4$). We follow \citet{adrian2019vulnerable}, and forecast average growth rates over the specific horizons.

To evaluate the performance of the different estimators, we will first conduct a small psuedo forecasting exercise similar to \citet{carriero2025specification}. This is followed by looking at the coefficient profiles of the different models to see whether we corroborate the finding of \citet{adrian2019vulnerable}, i.e. that lagged GDP growth acts as a location shifter, and NFCI is the key driver of nonlinearities in GDP growth. We will briefly look at the in-sample and out-of-sample weighted residuals for several values of $\alpha$ to gauge the bias-variance trade-off. We then turn to the policy implications of selecting models that incorporate interquantile shrinkage.

\subsection{Out-of-Sample Performance}
To evaluate out-of-sample performance of the estimators we will do a pseudo forecasting exercise similar to the one in \citet{carriero2025specification}. For the one- and four-quarter ahead forecast horizon we will compute the out-of-sample GDP densities on an expanding window, where the initial in-sample period uses the first 50 observations of the sample. This means that there are in total $150-h$ forecast windows to evaluate.

\subsubsection{Crossing incidence}
Before evaluating the out-of-sample performance, we note that having $\alpha=1$ only ensures non-crossing quantiles in-sample but not out-of-sample. However, when setting $\alpha>1$, the non-crossing constraints are tighter which leads to shrinking the fitted quantiles towards parallel lines. This in turn leads to less out-of-sample crossing as well. Since the FLQR also shrinks towards parallel fitted quantiles, one question is whether the FLQR also leads to less out-of-sample crossing. To this end we will compare the out-of-sample crossing incidence for all the estimators (QR, BRQ, GNCQR, FLQR). The crossing incidence is calculated by comparing the fitted quantiles with the sorted quantiles following the procedure of \citet{chernozhukov2010crossing}:

\begin{equation}
    CrossI=\frac{1}{Q(T-50)}\sum^T_{t=51}\sum^Q_{q=1}\textbf{I}[{\hat{\mathcal{Q}}({\tau_q,t+h|\mathcal{F}_t})=\hat{\mathcal{Q}}_{sort}({\tau_q,t+h|\mathcal{F}_t})}]
\end{equation}

\noindent where $\hat{\mathcal{Q}}({\tau_q,t+h|\mathcal{F}_t})$ is the forecasted quantile at time $t$ and forecast horizon $h$ and $\hat{\mathcal{Q}}_{sort}({\tau_q,t+h|\mathcal{F}_t})$ is the sorted forecast quantile. The crossing incidence measures what proportion of quantiles need sorting after estimation to obtain a valid CDF. The lower the $CrossI$ value the less quantiles need to be rearranged after estimation. The out-of-sample crossing incidence table is presented in table (\ref{tab:CrossI})

\begin{table}[]
    \centering
    \caption{Out-of-Sample Crossing Incidence for the different estimators}
    \label{tab:CrossI}
    \begin{tabular}{rcc}
\textit{} & \textbf{h=1} & \textbf{h=4} \\
\hline
$QR$ & 6.32\% & 5.62\% \\
$BRW$ & 1.37\% & 1.47\% \\
$FLQR$ & 0.91\% & 0.00\% \\
$GNCQR$ & 3.19\% & 1.79\% \\
\hline
\end{tabular}
\end{table}

The first thing to note is that more crossing occurs at the shorter forecast horizon than the longer one for all estimators. Second, the out-of-sample crossing incidence is the highest for the traditional quantile regression method. This indicates that the rearrangement algorithm of \citet{chernozhukov2010crossing} would have the most gains for the traditional quantile estimator. Note, that as mentioned by \citet{bondell2010noncrossing}, post-processing methods like the rearrangement method have the potential to improve fitted quantiles without changing the estimated coefficients. Third, GNCQR yields the lowest out-of-sample crossing incidence, even yielding proper densities for $h=4$. Finally, the FLQR estimator leads to worse out-of-sample crossing than the BRW estimator. Taken together these results indicate that the GNCQR is most likely to yield proper forecasted densities.

\subsubsection{Forecast Performance}
Given that the sorting procedure is widely used in forecasting settings, we will evaluate forecast performance of sorted and unsorted quantiles. We opt to look at both sorted and unsorted quantiles because of the results in table (\ref{tab:CrossI}) showing that sorting is likely to improve the performance of the QR and FLQR more than the BRW and GNCQR. To evaluate forecast performance, several measures are used. First, the quantile score (QS) is used with is the tick-loss weighted residual for a given forecast observation \citep{giacomini2005evaluation}:

\begin{equation}
    QS_{\tau_q,t+h}=(y_{t+h}-\hat{\mathcal{Q}}({\tau_q,t+h|\mathcal{F}_t}))(\tau_q-\textbf{I}[{|y_{t+h}\leq\hat{\mathcal{Q}}({\tau_q,t+h|\mathcal{F}_t})|}])
\end{equation}

\noindent where $\hat{\mathcal{Q}}({\tau_q,t+h|\mathcal{F}_t})$ is the forecasted quantile and $y_{t+h}$ is the unobserved value of GDP at time $t$.

To get an overall picture of density forecast performance we will follow \citet{knotek2019financial} and \citet{carriero2025specification} in using diffusion indices to gauge relative forecast performance of the different estimators. For our model comparison we will use the traditional quantile estimator (QR) to be the baseline model. Let $rQS_{i,\tau_q,h}$ be the ratio of the Quantile Score for estimator i, at quantile $\tau_q$, for forecast horizon $h=1,4$. Given the ratios of QS we calculate the case specific index as follow:

\begin{equation}
f_{i,\tau_q,h}=
    \begin{cases}
    1& \text{if } rQS\geq 1+s\\
    -1& \text{if } rQS\leq1-s\\
    0 & \text{otherwise}
\end{cases}
\end{equation}

\noindent where $s$ is the chosen sensitivity in relative performance. To obtain the diffusion index we take the average of $f_{i,\tau_q,h}$ over $\tau_q$, and $h$, i.e. $DI_i=N^{-1}\sum_{\forall\tau_q,h}f_{i,\tau_q,h}$. By construction $DI_i$ is between 1 and -1, with lower values indicating better performance relative to the QR forecasts. Values closer to 0 indicate no difference between the estimator and traditional quantile regression. Note, that the diffusion index is sensitive to the value of $s$ chosen. Setting $s$ to smaller values makes the performance index more sensitive to changes relative to the QR, while setting $s$ to larger values will only consider large improvements in forecasting performance. In \citet{carriero2025specification} $s=0.05$ while in \citet{knotek2019financial} $s=0.1$. Rather than pick one specific value for $s$ we will compare various sensitivity values: $s={0.005,0.01,0.05}$. The results for the diffusion indices for the various $s$ values for the unsorted and sorted. The results for the diffusion index are shown in table (\ref{tab:DI}). Note that in this table the reference model is always the respective QR model, i.e. for unsorted quantiles the reference is QS derived from the unsorted QR fits while for the sorted quantiles the reference is the QS derived from the sorted QR fits.

\begin{table}[]
\centering
\caption{Diffusion Index for the estimators at different sensitivities}
\label{tab:DI}
\begin{tabular}{l|ccc|ccc}
 & \multicolumn{3}{c|}{Unsorted quantiles} & \multicolumn{3}{c}{Sorted quantiles} \\
 & BRW & GNCQR & FLQR & BRW & GNCQR & FLQR \\ \hline
s=0.005 & -28.95\% & -36.84\% & 15.79\% & -15.79\% & -10.53\% & 31.58\% \\
s=0.01 & -13.16\% & -31.58\% & 5.26\% & 0.00\% & -7.89\% & 21.05\% \\
s=0.05 & -7.89\% & -5.26\% & 0.00\% & 0.00\% & 0.00\% & 2.63\% \\ \hline
\end{tabular}

\end{table}

The table reveals that the GNCQR and BRW has gains over the QR estimator regardless if the sorted or unsorted QS are used. The table also shows that much of the improvements are smaller given that the DI becomes less negative the larger $s$ is. For the unsorted quantiles, the GNCQR and BRW yield improvements over the QR even when setting $s=0.05$. The same cannot be said for the sorted quantiles, where $s$ has to be set to 0.01 to see any difference in performance between the sorted QR quantiles and the sorted GNCQR quantiles. Furthermore, the sorted BRW quantiles perform on par with the sorted QR. Note that from table (\ref{tab:CrossI}) we know that the GNCQR improves the least from sorting given it has the lowest crossing incidence. As such, it is reassuring that even in the sorted quantiles, the GNCQR provides gains over sorted quantiles corroborating the fit results from our Monte Carlo exercise. The table also shows that the FLQR yields no improvements in quantiles score compared to QR regardless of what $s$ is set and whether we sort the quantiles. In essence, these findings add to \citet{carriero2025specification}, namely that fused shrinkage is also important for growth-at-risk. Importantly, our results indicate that perhaps just as important is how this interquantile shrinkage is induced: adaptive non-crossing constraints yield improvements, while the traditional fused LASSO seem to lead to worse forecasting performance.

While the diffusion index of the QS gives a broad view of how the estimator performs, it doesn't inform us about what part of the distribution the estimator does better (or worse). To this end we also construct quantile weighted CRPS (qwCRPS) scores of \citet{gneiting2011comparing} for all the estimators for the different forecast horizons. To calculate these measures, we take the QS and apply quantile specific weights:

\begin{equation}
    qwCRPS_{t+h} = \int^1_0 \; w_{\tau_q} QS_{t+h,\tau_q}d\tau_q,
\end{equation}

\noindent where $w_{\tau_q}$ denotes a weighting scheme to evaluate specific parts of the forecast density. We choose this measure as the scoring rule because through different weighting schemes we can evaluate differences at different parts of the distribution. We consider 3 different weighting schemes: $w_{\tau_q}^1=\frac{1}{Q}$ places equal weight on all quantiles (denoted as CRPS), which is equivalent to taking the average of the weighted residuals; $w_{\tau_q}^2={\tau_q}(1-{\tau_q})$ places more weight on central quantiles; $w_{\tau_q}^3=(1-{\tau_q})^2$ places more weight on the left tail; and $w_{\tau_q}^4={\tau_q}^2$ places more weight on the right tail. The results using the different weighting schemes are presented in Table (\ref{tab:ForcRes3}). 

\begin{table}[t]
\centering
\caption{Forecast results of the different estimators}
\label{tab:ForcRes3}
\begin{tabular}{lrcccc|cccc}
\textbf{} & \textit{} & \multicolumn{4}{c|}{Unsorted} & \multicolumn{4}{c}{Sorted} \\
\textbf{} & \textit{} & CRPS & Centre & Left & Right & CRPS & Centre & Left & Right \\ \hline
\multicolumn{2}{l}{\textbf{h=1}} &  &  &  &  &  &  \\
\textbf{} & \textit{QR} & 0.916 & 0.173 & 0.280 & 0.463 & 0.909 & 0.172 & 0.277 & 0.459 \\
\textbf{} & \textit{BRW} & 0.914 & 0.172 & 0.279 & 0.462 & 0.911 & 0.173 & 0.278 & 0.461 \\
\textbf{} & \textit{GNCQR} & 0.910 & 0.172 & 0.274 & 0.464 & 0.909 & 0.172 & 0.274 & 0.463 \\
\textbf{} & \textit{FLQR} & 0.918 & 0.173 & 0.279 & 0.466 & 0.914 & 0.173 & 0.277 & 0.463 \\ \hline
\multicolumn{2}{l}{\textbf{h=4}} &  &  &  &  &  &  \\
\textbf{} & \textit{QR} & 0.559 & 0.108 & 0.168 & 0.284 & 0.554 & 0.107 & 0.167 & 0.280 \\
\textbf{} & \textit{BRW} & 0.554 & 0.107 & 0.166 & 0.281 & 0.553 & 0.107 & 0.166 & 0.280 \\
\textbf{} & \textit{GNCQR} & 0.552 & 0.107 & 0.165 & 0.280 & 0.552 & 0.107 & 0.165 & 0.280 \\
\textbf{} & \textit{FLQR} & 0.560 & 0.108 & 0.167 & 0.285 & 0.557 & 0.108 & 0.166 & 0.284 \\ \hline
\end{tabular}
\end{table}

Looking at the unsorted results of Table \ref{tab:ForcRes3}, we can see that GNCQR yields the lowest weighted residual for all forecast horizons, for all weighting schemes except the right tail. These results indicate that the gains in forecast performance in GNCQR are largely from better left tail and centre forecasts. Given that the growth-at-risk framework was designed to capture downside (left tail) risk, the inability to yield better right side forecasts for the short forecast horizon is less of a concern. These results also show how BRW is a strong candidate to obtain density estimates. In particular, it is often only beaten by GNCQR. This highlights the usefulness of BRW in density estimation and for most applications should suffice as a first candidate to estimate.

Overall proposed method provides gains in forecasting. This highlights how even in small (covariate) dimensional settings, one can have better forecast performance through correctly identifying which variables are quantile varying.

\subsection{Coefficients}
While quantile regression is robust to outliers, extreme values in the covariates can still cause issues. Since the GaR is essentially a quantile autoregressive model of \cite{koenker2006quantile}, the extreme GDP observations of the COVID period enter the independent variables through the lag. This could potentially have undue influence on the coefficients and as such we will run the model on two sample periods: Pre-Covid and Full sample.

\begin{figure}[t]
    \centering
    \includegraphics[width=0.85\textwidth]{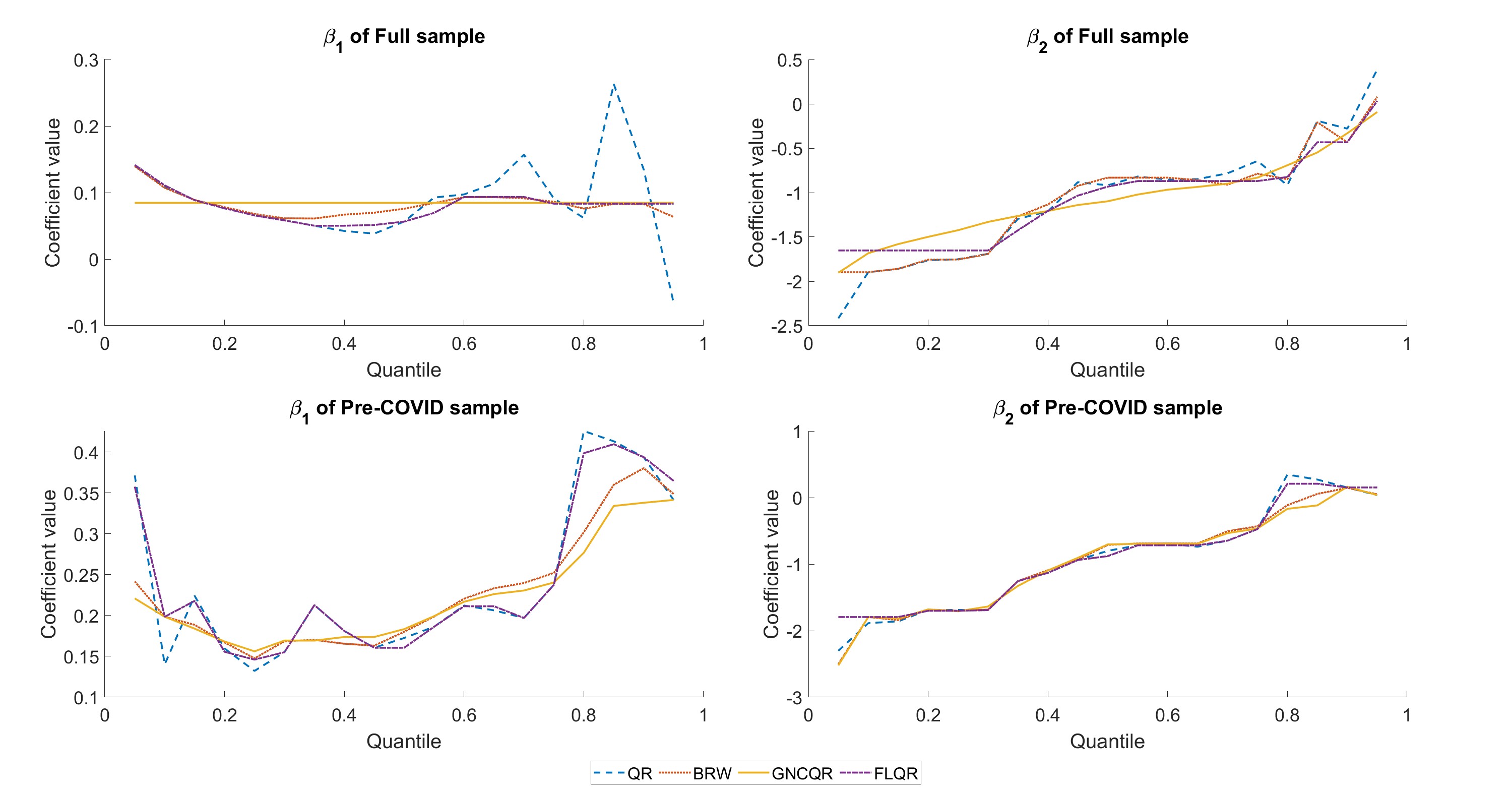}
    \caption{$\beta$ coefficients for the different estimators at h=1}
    \label{fig:Betah1}
\end{figure}

\begin{figure}[h!]
    \centering
    \includegraphics[width=0.85\textwidth]{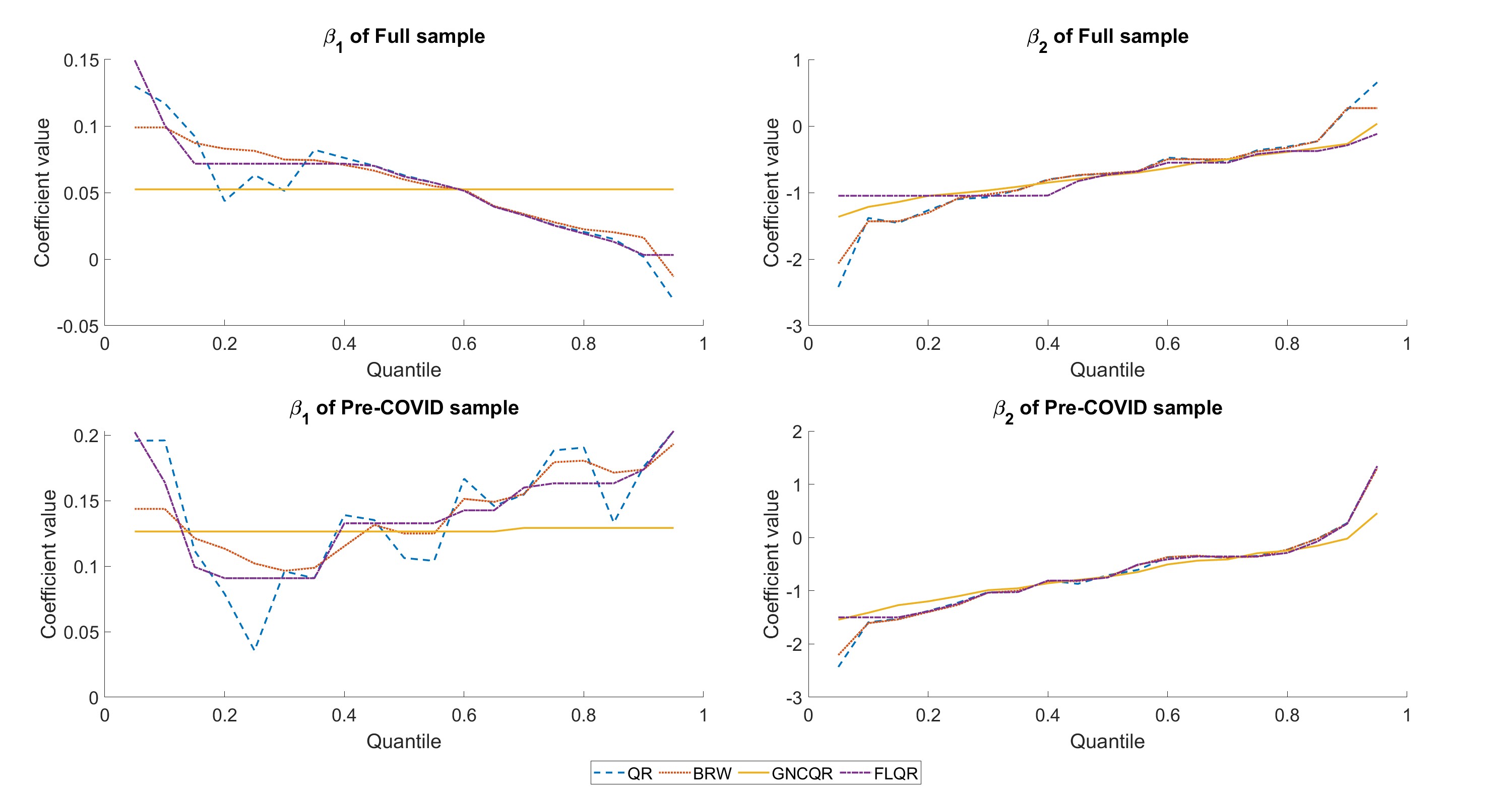}
    \caption{$\beta$ coefficients for the different estimators at h=4}
    \label{fig:Betah4}
\end{figure}

The quantile profiles of the estimated $\beta$ coefficients are presented in Figures \ref{fig:Betah1} and \ref{fig:Betah4}, for pre-COVID and full samples. $\beta_1$ is the coefficient for GDP growth and $\beta_2$ corresponds to the NFCI. A key observation is that the coefficient for the NFCI has a distinct quantile-varying profile for both samples and forecast horizons. Almost all estimators showcase a general increasing tendency across the quantiles. Nevertheless, there are some difference between the different estimators. First, FLQR is likely to shrink quantile variation at the tails, particularly the lower tail. This tendency sees FLQR have a less pronounced impact of NFCI on the lowest quantiles than the other estimators. Second, the traditional QR estimator has a positive coefficient for the upper quantiles of NFCI, while GNCQR is more likely to give a zero coefficient. Given the argument of \citet{adrian2019vulnerable}, a zero coefficient for the upper quantiles of NFCI makes more sense than positive coefficients. Note that none of the estimators have shrinkage imposed on the $\beta$, i.e. there is no explicit variable selection, yet GNCQR obtains $\hat{\beta}$'s much closer  to 0 at the upper quantile of NFCI. Third, GNCQR has less `jagged' quantile profile. In fact, looking at the NFCI coefficients quantile profile, GNCQR only has negative slope for the pre-covid sample for the $h=1$ forecast horizon. For all the other cases, GNCQR portrays a gentle upwards sloping profile.

While there is significant quantile variation for the coefficient of NFCI, the same cannot be said for the coefficient of GDP growth. The GNCQR yields virtually no quantile variation in almost all cases, with only pre-COVID of $h=1$ showing some variation. The other estimators also show less quantile variation for GDP growth than the NFCI, but they do not fully smooth out spurious quantile variation. Interestingly, FLQR also does not shrink away the quantile variation in GDP growth. 

The coefficient profile figures demonstrate that GNCQR is capable of identifying quantile variation better than simple FLQR. Furthermore, for the variables that have quantile variation, GNCQR shrinks away the ``jagged" edges resulting in smoother quantile profiles. This is likely the reason for the GNCQR's superior forecast performance: the GNCQR is able to identify quantile varying variables and shrink away any unnecessary variation in these variables profiles. 

These results verify the findings of \citet{adrian2019vulnerable}, namely that macro-financial linkages are important for downside risk of GDP growth. Notably, the GNCQR completely shrinks the quantile variation in the GDP coefficient, which signifies that, while past values of GDP are important in determining the location of the GDP distribution, most variation in the shape of the distribution is on account of changes in the financial conditions. Furthermore, we add to the finding of \citet{carriero2025specification}: beyond shrinkage, interquantile shrinkage has benefits for empirical macroeconomics. This is particularly important if the goal is to obtain the coefficient profiles, since post-processing sorting procedures are applicable to the fitted quantiles and not the coefficient profiles.

\subsection{Bias-variance trade-off}
Figures (\ref{fig:OOSvsISh1}) and (\ref{fig:OOSvsISh4}) show the in-sample and out-of-sample average quantile score for different $\alpha$ values and compares the profiles with the fused LASSO estimator with the same hyperparameter value. 

\begin{figure}[t]
     \centering
     \begin{subfigure}[b]{0.8\textwidth}
         \centering
         \includegraphics[width=\textwidth]{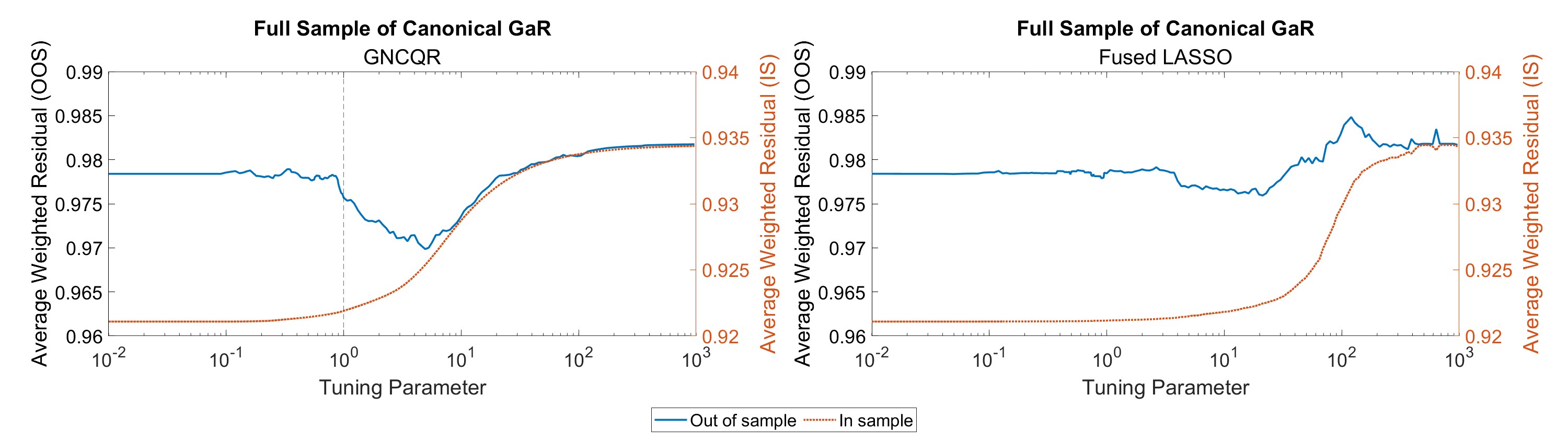}
     \end{subfigure}
     \vfill
     \begin{subfigure}[b]{0.8\textwidth}
         \centering
         \includegraphics[width=\textwidth]{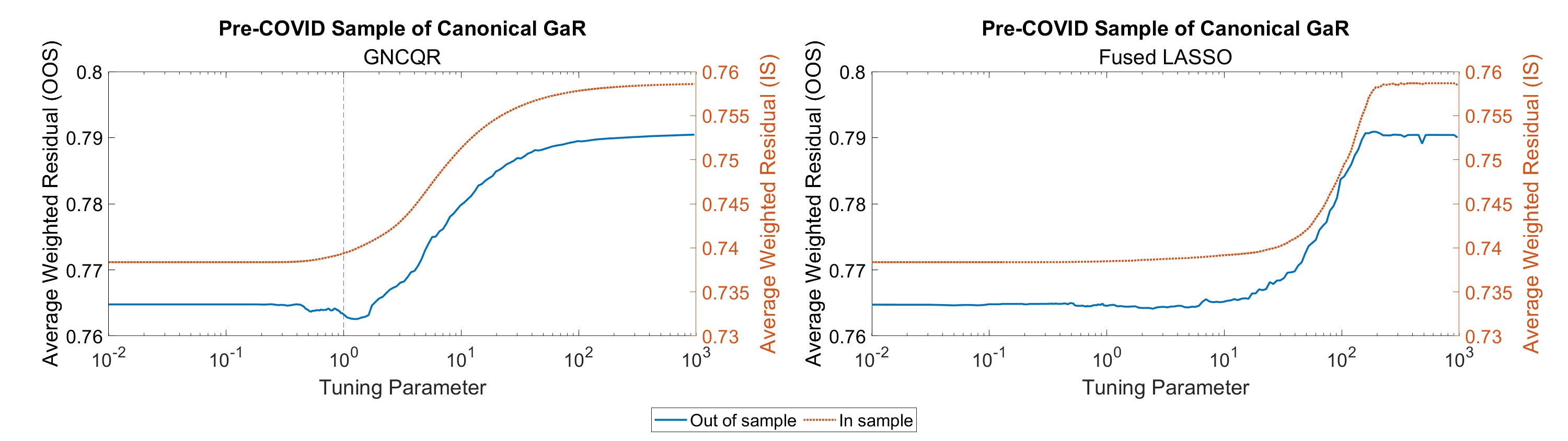}
     \end{subfigure}
     \caption{In-sample vs. Out-of-sample fit for various hyperparameter values ($h=1$)}
     \label{fig:OOSvsISh1}
\end{figure}
\begin{figure}[h!]
     \centering
     \begin{subfigure}[b]{0.8\textwidth}
         \centering
         \includegraphics[width=\textwidth]{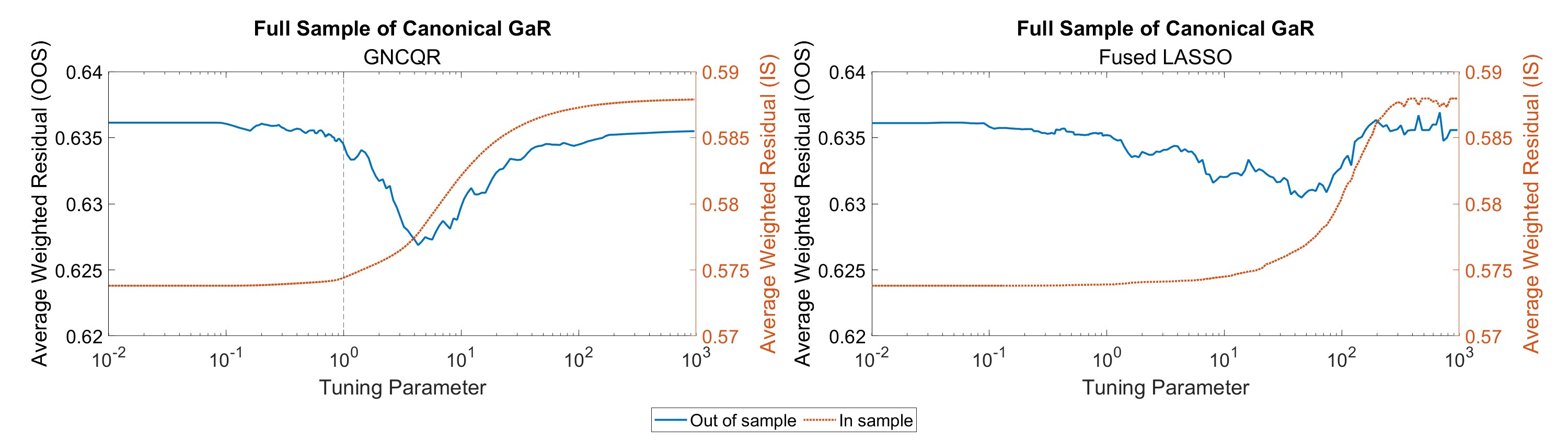}
     \end{subfigure}
     \vfill
     \begin{subfigure}[b]{0.8\textwidth}
         \centering
         \includegraphics[width=\textwidth]{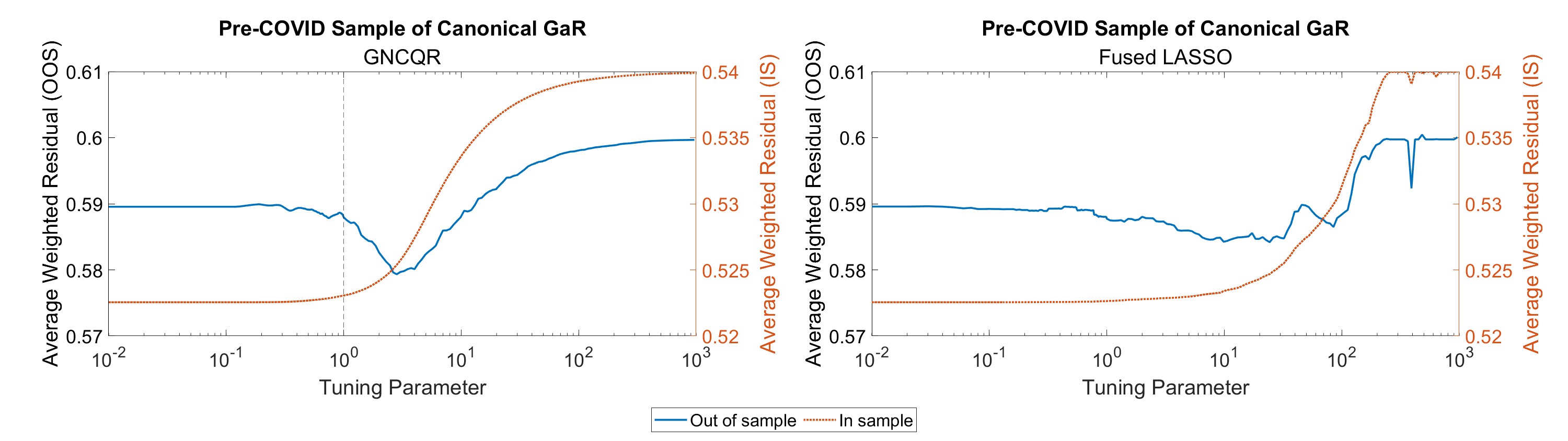}
     \end{subfigure}
     \caption{In-sample vs. Out-of-sample fit for various hyperparameter values ($h=4$)}
     \label{fig:OOSvsISh4}
\end{figure}

The figures show that the equivalent hyperparameters of the different estimators do not lead to the same values. This in turn leads to different profiles. In particular, GNCQR shows a more gradual change in in-sample fit for all horizons and subsamples. Furthermore, the out-of-sample fits of the GNCQR have a more pronounced minimum. Another feature of the GNCQR is that its in-sample and out-of-sample profiles are less "jagged" than that of the FLQR. This is true for all forecast horizons, as well as pre- and post-COVID. These are important for model selection: it is easier to optimise the tuning parameter selection when the minimum is obtain gradually. Furthermore, these properties indicate that the GNCQR shrinks quantile variation of the correct variable.

\subsection{Policy implication}
The previous section highlighted the importance of fused shrinkage for forecasting and getting coefficient profiles that have less spurious quantile variation. In this section we will highlight how these differences in the estimators could impact policy. 

\subsubsection{Expected Shortfall}
The first policy tool we will investigate is expected shortfall as outlined in \citet{adrian2019vulnerable}, which measures the total probability mass that the conditional distribution assigns to the left tail:

\begin{equation}
    SF_{t+h}=\frac{1}{\pi}\int^\pi_0\hat{F}^{-1}_{y_{t+h|\mathcal{F}_t}}(\tau|\mathcal{F}_t)d\tau
\end{equation}

To construct expected shortfall from the estimated conditional quantiles we will utilise the methodology of \citet{mitchell2022constructing}, which uses a nonparametric method to construct densities from estimated quantiles. We will construct the expected shortfall for the $5^{th}$ quantile using the nonparametric methodology. The constructed expected shortfall are in figure (\ref{fig:ES}). While the same procedure can be used to estimate Expected Longrise, we opt to only focus on Expected Shortfall, as the variables used in the GaR model are more likely to capture downside risk.

\begin{figure}[t]
     \centering
     \begin{subfigure}[b]{0.7\textwidth}
         \centering
         \includegraphics[width=\textwidth]{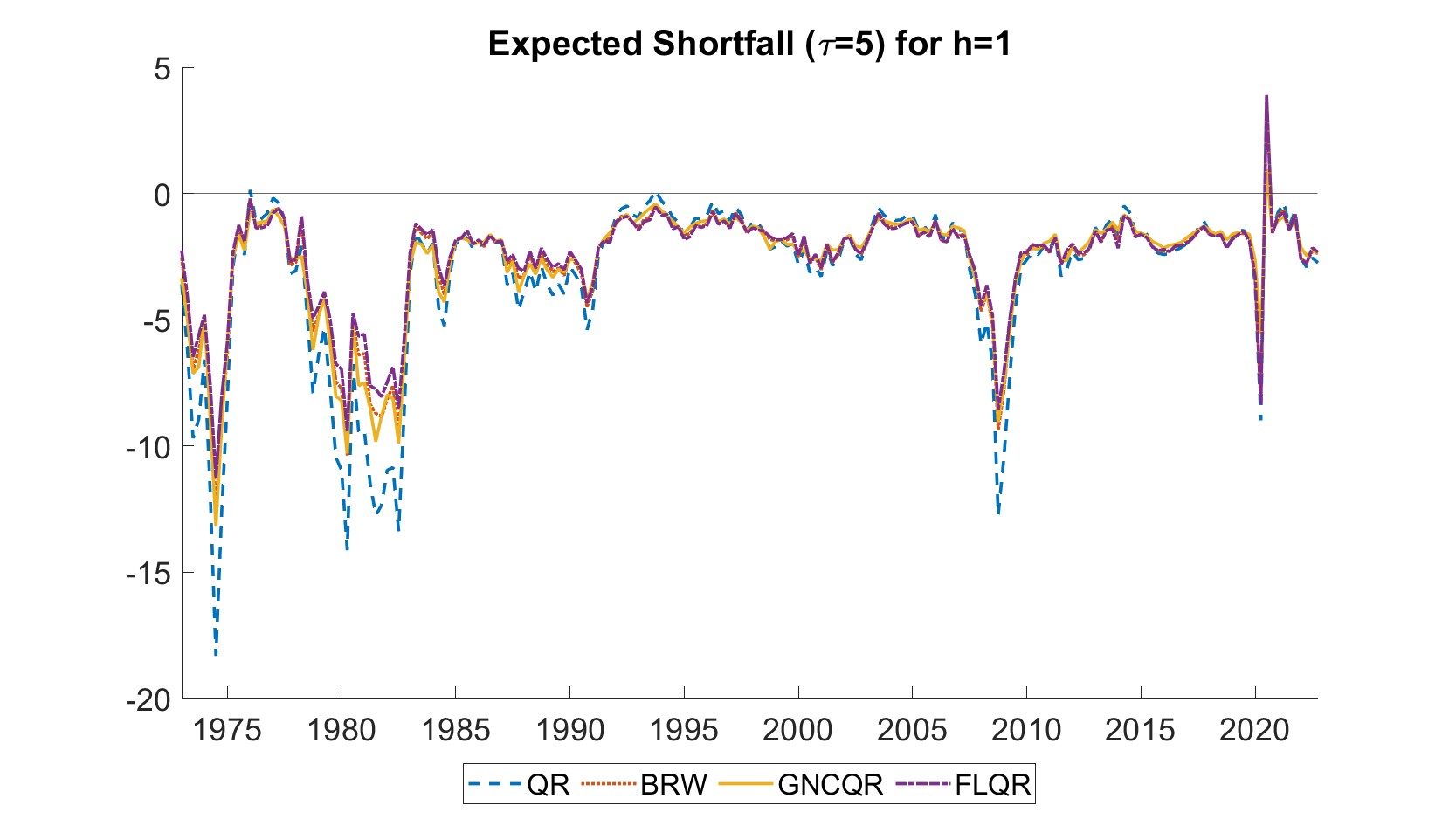}
     \end{subfigure}
     \vfill
     \begin{subfigure}[b]{0.7\textwidth}
         \centering
         \includegraphics[width=\textwidth]{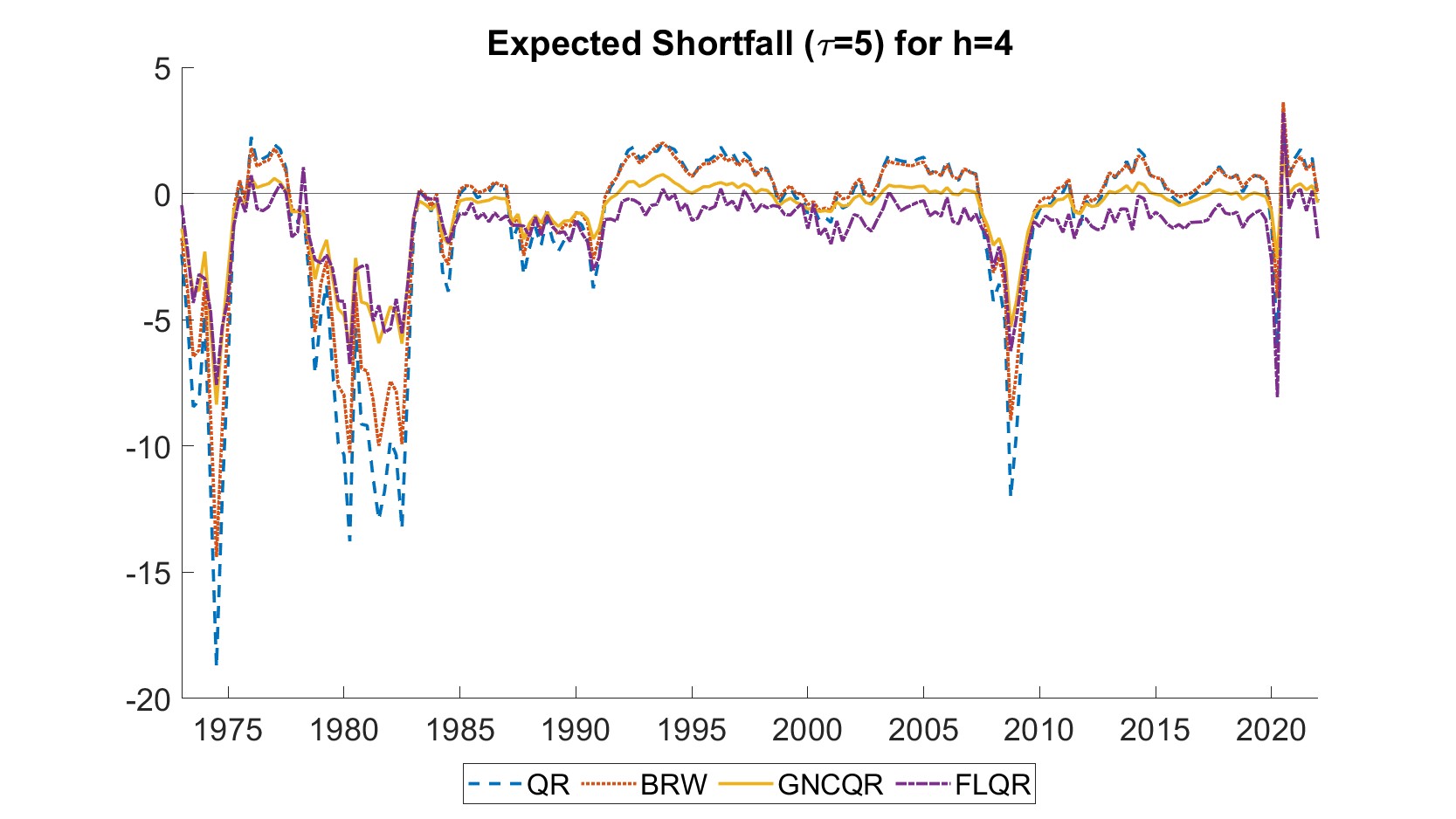}
     \end{subfigure}
     \caption{Expected Shortfall for the different estimators}
     \label{fig:ES}
\end{figure}

The figure reveals that there is a stark difference between the estimators for ES, especially at the one year horizon (h=4). These ES estimates are critical for policy-makers as they give an overview for the evolution of downside risk.

Looking at the difference between the estimators, we can see that the QR estimator produces the most volatile ES measure compared to the other estimators, especially at the start of the estimation period. Compared to QR, the BRW estimator produces smoother ES over time, while still dipping during extreme financial stress: namely the global financial crisis and the 1970s.

FLQR deviates the most from the QR and BRW, with producing more optimistic ES during financial stress episodes, but more pessimistic ES during more tranquil periods. This tendency is particularly clear for the one year horizon. Notably, FLQR shows consistently more negative ES estimates in this horizon during relatively stable market periods (e.g., 1990-2005), suggesting a potential pessimistic bias for the FLQR.

Just like with the variable selection results, the ES made from GNCQR coefficients yields a middle-ground: between the QR's ES and the FLQR's ES. In particular, the GNCQR ES is closer to the QR during financial stress episodes, yielding lower ES estimates. But during tranquil periods the method yields close to 0 ES values especially at the longer forecast horizon. Importantly, GNCQR yields smoother ES dynamics than QR. These advantages of GNCQR are likely on account of its ability to shrink away spurious quantile variation, which helps in identifying variables that drive downside risk in genuine market stress episodes.

The difference between QR and the other estimators, particularly during the early sample period, highlights the value of identifying variables that drive quantile variation. The smoother ES at longer horizons is particularly useful for policy makers. GNCQR can effectively temper the extreme volatility of traditional QR while maintaining sensitivity to genuine downside risks, potentially leading to more reliable and practical risk assessments for financial stability purposes.

\subsubsection{Tail local projection}
Another important policy tool for GaR is local projection. \citet{ruzicka2021quantile} notes that there is an identification assumption embedded in \citet{adrian2019vulnerable}, namely that NFCI has no contemporaneous effects on GDP growth distribution. This assumption has been used by \citet{ruzicka2021quantile}, \citet{wojciechowski2024structural}, and \citet{chavleishvili2024forecasting} to create quantile IRFs and quantile local projections. Importantly, these methods rely on the estimated coefficients to construct the quantile projections through time, which cannot be sorted. We will produce quantile local projections for the 4 estimators for the $5^{th}$, $10^{th}$, and $25^{th}$ quantile. We will follow \citet{ruzicka2021quantile} and deviate from \citet{adrian2019vulnerable} on two fronts: a) we will estimate future growth directly, rather than averaged over the forecast horizon; and b) we will include additional lags of GDP and NFCI. The quantile LPs of the response of GDP to a unit shock in NFCI are shown in figure (\ref{fig:LPQR}).

\begin{figure}
    \centering
    \includegraphics[width=0.6\linewidth]{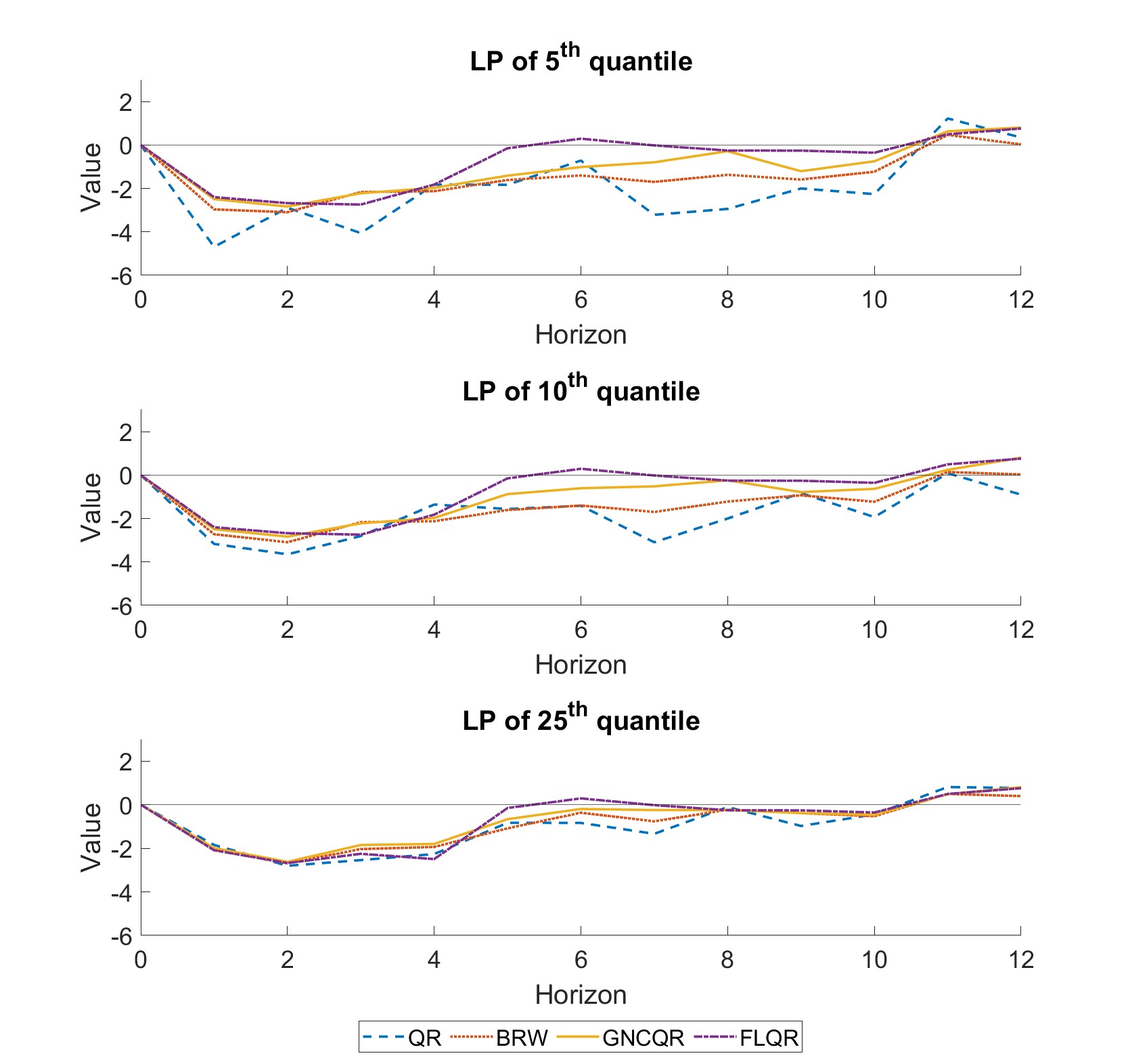}
    \caption{Local projections for the different estimators at different quantiles}
    \label{fig:LPQR}
\end{figure}

Across all three quantiles, there's a general pattern of initial decline followed by recovery, but the timing and magnitude of these movements vary considerably by estimator. The QR consistently shows the deepest and most volatile projections, while FLQR shows the quickest recovery. BRW and GNCQR typically fall between these extremes, with GNCQR often showing slightly less negative projections than BRW, particularly in the medium run (h=4-8).

In the extreme lower tail (5th quantile), we observe substantial heterogeneity across estimators. The QR estimator exhibits the highest volatility, with deeper initial drops reaching below -4 around h=1-2 and again near h=6-7. Indeed this volatility was the primary motivation for \citet{ruzicka2021quantile} to propose smoothing the projections dynamically over the horizons. In contrast, the FLQR estimator shows a remarkably different pattern, returning to near-zero values much more quickly after horizon 4, indicating it might be producing more conservative estimates of tail risk. When looking at the coefficient profiles we saw that the FLQR was prone to shrinking quantile variation at the tails. As such it is likely that the quick recovery suggested by the FLQR is a consequence of this overshrinking behaviour. The GNCQR and BRW produce smoother projections than the QR but do not return to 0 as quickly as the FLQR. This finding is particularly interesting and suggests that potentially the smoothed quantile LP framework of \citet{ruzicka2021quantile} induced fused shrinkage at each horizon. We leave for future research to explore the link between smoothing quantile variation and smoothing projections over time.

The differences between estimator choices diminish as we move up the quantiles, with the 25th quantile showing limited difference between the estimated profiles especially in the short run. After h=4, the estimators begin to diverge again, though less dramatically than in the lower quantiles.

These differences highlight the importance of estimator choice in growth-at-risk analyses, especially when focusing on extreme quantiles that are of particular interest to policymakers concerned with tail risks to economic growth. The varying persistence and magnitude of negative projections across estimators could lead to significantly different assessments of downside risks and potentially different policy recommendations.


\section{Conclusion}

This paper develops an adaptive non-crossing constraint, encompassing the traditional quantile regression estimator, the \citet{bondell2010noncrossing} non-crossing estimator, as well as the composite quantile regression, as special cases by varying tightness of the constraint. By developing non-crossing constraints that can be tightened, we study the properties of these constraints on the estimated $\beta$ parameters. Doing so reveals that non-crossing constraints are simply a type of Fused LASSO with quantile specific shrinkage parameters.

We verify that imposing non-crossing constraints is equivalent to introducing fused shrinkage with a Monte Carlo experiment. We also show how our proposed Generalised Non-Crossing Quantile Regression (GNCQR) estimator provides fits that are either better or nearly as good as those of the BRW estimator \citep{bondell2010noncrossing}. When looking at the variable selection properties of the different estimators we found that traditional Fused LASSO (FLQR) estimator overshrinks while BRW estimator undershrinks compared to GNCQR.


Taking an interquantile shrinkage lens to US Growth-at-Risk, we found that our proposed GNCQR estimator outperforms both traditional quantile regression and standard fused LASSO approaches, particularly in forecasting left-tail risks. The GNCQR effectively identifies variables that drive quantile variation while shrinking away spurious variation, resulting in smoother coefficient profiles and improved forecast performance. Notably, our results confirm the findings of \citet{adrian2019vulnerable} that financial conditions are the drivers of nonlinearities in GDP distribution, while past GDP growth acts primarily as a location shifter.

The policy implications of better GaR coefficients is also explored. Expected Shortfall estimates derived from GNCQR exhibit less volatility while remaining sensitive to genuine financial stress episodes, offering policymakers a more reliable assessment of downside risks. Similarly, quantile local projections based on GNCQR provide balanced medium-term forecasts that avoid both the excessive volatility of traditional quantile regression and the potentially over-optimistic recovery paths suggested by standard fused LASSO. These improvements in policy-relevant metrics underscore the importance of appropriate interquantile shrinkage in macroeconomic risk assessment.


In sum, the non-crossing constraints of Equation (\ref{eq:ADAconstraint}) bestow upon the quantile estimator additional attractive properties: (1) it can distinguish variables with quantile variation from those without, and shrinks the correct variable; and (2) it renders the estimated quantile profiles less `jagged', i.e., it removes sudden reversions in the difference in the $\beta$ coefficient.

The findings of this paper reach beyond the estimator proposed here. In particular, because of the equivalence between non-crossing and fused-shrinkage, one can extend the methodology to obtain non-crossing Bayesian quantile regression estimators like in \citet{lancaster2010bayesian}. To our knowledge, currently the most popular way to obtain non-crossing Bayesian quantiles involve post-processing methods such as \citet{rodrigues2017regression}. Implementing theorem (\ref{theorem:NC=FLASSO}) to the Bayesian realm has clear advantages.

The formulation of the GNCQR is very simple: the hyperparameter is simply a scalar. It is probable that one can gain further improvements in estimation by making the hyperparameter have more than just one dimension: making the hyperparameter variable specific has been shown to offer better performance in other contexts \citep{zou2006adaptive}. Naturally, doing so will require a better hyperparameter tuning procedure, since grid-search is a computation-intensive choice even with just a single parameter. We leave this extension to future research.

Another avenue of research is shrinking both the level and difference coefficients, similar to \cite{jiang2014interquantile}. However, following \citet{szendrei2023revisiting}, shrinkage in levels would require a separate dedicated hyperparameter. This way the two types of shrinkage could be incorporated yielding the possibility to identify quantile varying sparsity like in \citet{kohns2021decoupling} and \citet{szendrei2023revisiting}, while being principled about the degree of fused shrinkage one needs to impose. We leave this work for future research. 


Finally, since the bias introduced by non-crossing constraints is similar to biased bootstrap methods proposed for regression under monotone or convexity constraints, one can use degree of bias to test model adequacy. Developing tests of model adequacy using bias induced by the quantile non-crossing constraints is a promising avenue for future research.

\pagebreak

\bibliographystyle{chicago}
\bibliography{text.bbl}

\begin{thebibliography}{}

\bibitem[\protect\citeauthoryear{Adrian, Boyarchenko, and Giannone}{Adrian et~al.}{2019}]{adrian2019vulnerable}
Adrian, T., N.~Boyarchenko, and D.~Giannone (2019).
\newblock Vulnerable growth.
\newblock {\em American Economic Review\/}~{\em 109\/}(4), 1263--89.

\bibitem[\protect\citeauthoryear{Adrian, Boyarchenko, and Giannone}{Adrian et~al.}{2021}]{adrian2021multimodality}
Adrian, T., N.~Boyarchenko, and D.~Giannone (2021).
\newblock Multimodality in macrofinancial dynamics.
\newblock {\em International Economic Review\/}~{\em 62\/}(2), 861--886.

\bibitem[\protect\citeauthoryear{Bates, Hastie, and Tibshirani}{Bates et~al.}{2024}]{bates2024cross}
Bates, S., T.~Hastie, and R.~Tibshirani (2024).
\newblock Cross-validation: what does it estimate and how well does it do it?
\newblock {\em Journal of the American Statistical Association\/}~{\em 119\/}(546), 1434--1445.

\bibitem[\protect\citeauthoryear{Bergstra and Bengio}{Bergstra and Bengio}{2012}]{bergstra2012random}
Bergstra, J. and Y.~Bengio (2012).
\newblock Random search for hyper-parameter optimization.
\newblock {\em Journal of machine learning research\/}~{\em 13\/}(2).

\bibitem[\protect\citeauthoryear{Bondell, Reich, and Wang}{Bondell et~al.}{2010}]{bondell2010noncrossing}
Bondell, H.~D., B.~J. Reich, and H.~Wang (2010).
\newblock Noncrossing quantile regression curve estimation.
\newblock {\em Biometrika\/}~{\em 97\/}(4), 825--838.

\bibitem[\protect\citeauthoryear{Carriero, Clark, and Marcellino}{Carriero et~al.}{2025}]{carriero2025specification}
Carriero, A., T.~E. Clark, and M.~Marcellino (2025).
\newblock Specification choices in quantile regression for empirical macroeconomics.
\newblock {\em Journal of Applied Econometrics\/}~{\em 40\/}(1), 57--73.

\bibitem[\protect\citeauthoryear{Cerqueira, Torgo, and Mozeti{\v{c}}}{Cerqueira et~al.}{2020}]{cerqueira2020evaluating}
Cerqueira, V., L.~Torgo, and I.~Mozeti{\v{c}} (2020).
\newblock Evaluating time series forecasting models: An empirical study on performance estimation methods.
\newblock {\em Machine Learning\/}~{\em 109}, 1997--2028.

\bibitem[\protect\citeauthoryear{Chavleishvili and Manganelli}{Chavleishvili and Manganelli}{2024}]{chavleishvili2024forecasting}
Chavleishvili, S. and S.~Manganelli (2024).
\newblock Forecasting and stress testing with quantile vector autoregression.
\newblock {\em Journal of Applied Econometrics\/}~{\em 39\/}(1), 66--85.

\bibitem[\protect\citeauthoryear{Chernozhukov, Fernandez-Val, and Galichon}{Chernozhukov et~al.}{2009}]{chernozhukov2009improving}
Chernozhukov, V., I.~Fernandez-Val, and A.~Galichon (2009).
\newblock Improving point and interval estimators of monotone functions by rearrangement.
\newblock {\em Biometrika\/}~{\em 96\/}(3), 559--575.

\bibitem[\protect\citeauthoryear{Chernozhukov, Fernández-Val, and Galichon}{Chernozhukov et~al.}{2010}]{chernozhukov2010crossing}
Chernozhukov, V., I.~Fernández-Val, and A.~Galichon (2010).
\newblock Quantile and probability curves without crossing.
\newblock {\em Econometrica\/}~{\em 78\/}(3), 1093--1125.

\bibitem[\protect\citeauthoryear{Figueres and Jaroci{\'n}ski}{Figueres and Jaroci{\'n}ski}{2020}]{figueres2020vulnerable}
Figueres, J.~M. and M.~Jaroci{\'n}ski (2020).
\newblock Vulnerable growth in the euro area: Measuring the financial conditions.
\newblock {\em Economics Letters\/}~{\em 191}, 109126.

\bibitem[\protect\citeauthoryear{Giacomini and Komunjer}{Giacomini and Komunjer}{2005}]{giacomini2005evaluation}
Giacomini, R. and I.~Komunjer (2005).
\newblock Evaluation and combination of conditional quantile forecasts.
\newblock {\em Journal of Business \& Economic Statistics\/}~{\em 23\/}(4), 416--431.

\bibitem[\protect\citeauthoryear{Gneiting and Ranjan}{Gneiting and Ranjan}{2011}]{gneiting2011comparing}
Gneiting, T. and R.~Ranjan (2011).
\newblock Comparing density forecasts using threshold-and quantile-weighted scoring rules.
\newblock {\em Journal of Business \& Economic Statistics\/}~{\em 29\/}(3), 411--422.

\bibitem[\protect\citeauthoryear{Iseringhausen, Petrella, and Theodoridis}{Iseringhausen et~al.}{2023}]{iseringhausen2023aggregate}
Iseringhausen, M., I.~Petrella, and K.~Theodoridis (2023).
\newblock Aggregate skewness and the business cycle.
\newblock {\em Review of Economics and Statistics\/}, 1--37.

\bibitem[\protect\citeauthoryear{Jiang, Bondell, and Wang}{Jiang et~al.}{2014}]{jiang2014interquantile}
Jiang, L., H.~D. Bondell, and H.~J. Wang (2014).
\newblock Interquantile shrinkage and variable selection in quantile regression.
\newblock {\em Computational statistics \& data analysis\/}~{\em 69}, 208--219.

\bibitem[\protect\citeauthoryear{Jiang, Wang, and Bondell}{Jiang et~al.}{2013}]{jiang2013interquantile}
Jiang, L., H.~J. Wang, and H.~D. Bondell (2013).
\newblock Interquantile shrinkage in regression models.
\newblock {\em Journal of Computational and Graphical statistics\/}~{\em 22\/}(4), 970--986.

\bibitem[\protect\citeauthoryear{Jones}{Jones}{1994}]{jones1994expectiles}
Jones, M.~C. (1994).
\newblock Expectiles and m-quantiles are quantiles.
\newblock {\em Statistics \& Probability Letters\/}~{\em 20\/}(2), 149--153.

\bibitem[\protect\citeauthoryear{Knotek and Zaman}{Knotek and Zaman}{2019}]{knotek2019financial}
Knotek, E.~S. and S.~Zaman (2019).
\newblock Financial nowcasts and their usefulness in macroeconomic forecasting.
\newblock {\em International Journal of Forecasting\/}~{\em 35\/}(4), 1708--1724.

\bibitem[\protect\citeauthoryear{Koenker}{Koenker}{1984}]{koenker1984note}
Koenker, R. (1984).
\newblock A note on l-estimates for linear models.
\newblock {\em Statistics \& probability letters\/}~{\em 2\/}(6), 323--325.

\bibitem[\protect\citeauthoryear{Koenker}{Koenker}{2005}]{koenker2005}
Koenker, R. (2005).
\newblock {\em Quantile regression}.
\newblock New York: Cambridge University Press.

\bibitem[\protect\citeauthoryear{Koenker and Bassett}{Koenker and Bassett}{1978}]{koenker1978regression}
Koenker, R. and G.~Bassett (1978).
\newblock Regression quantiles.
\newblock {\em Econometrica: journal of the Econometric Society\/}, 33--50.

\bibitem[\protect\citeauthoryear{Koenker and Xiao}{Koenker and Xiao}{2006}]{koenker2006quantile}
Koenker, R. and Z.~Xiao (2006).
\newblock Quantile autoregression.
\newblock {\em Journal of the American statistical association\/}~{\em 101\/}(475), 980--990.

\bibitem[\protect\citeauthoryear{Kohns and Szendrei}{Kohns and Szendrei}{2021}]{kohns2021decoupling}
Kohns, D. and T.~Szendrei (2021).
\newblock Decoupling shrinkage and selection for the bayesian quantile regression.
\newblock {\em arXiv preprint arXiv:2107.08498\/}.

\bibitem[\protect\citeauthoryear{Kohns and Szendrei}{Kohns and Szendrei}{2024}]{kohns23hsbqr}
Kohns, D. and T.~Szendrei (2024).
\newblock Horseshoe prior bayesian quantile regression.
\newblock {\em Journal of the Royal Statistical Society Series C: Applied Statistics\/}~{\em 73\/}(1), 193--220.

\bibitem[\protect\citeauthoryear{Korobilis}{Korobilis}{2017}]{korobilis2017quantile}
Korobilis, D. (2017).
\newblock Quantile regression forecasts of inflation under model uncertainty.
\newblock {\em International Journal of Forecasting\/}~{\em 33\/}(1), 11--20.

\bibitem[\protect\citeauthoryear{Lancaster and Jae~Jun}{Lancaster and Jae~Jun}{2010}]{lancaster2010bayesian}
Lancaster, T. and S.~Jae~Jun (2010).
\newblock Bayesian quantile regression methods.
\newblock {\em Journal of Applied Econometrics\/}~{\em 25\/}(2), 287--307.

\bibitem[\protect\citeauthoryear{Liu and Wu}{Liu and Wu}{2009}]{liu2009stepwise}
Liu, Y. and Y.~Wu (2009).
\newblock Stepwise multiple quantile regression estimation using non-crossing constraints.
\newblock {\em Statistics and its Interface\/}~{\em 2\/}(3), 299--310.

\bibitem[\protect\citeauthoryear{Mitchell, Poon, and Zhu}{Mitchell et~al.}{2024}]{mitchell2022constructing}
Mitchell, J., A.~Poon, and D.~Zhu (2024).
\newblock Constructing density forecasts from quantile regressions: Multimodality in macrofinancial dynamics.
\newblock {\em Journal of Applied Econometrics\/}~{\em 39\/}(5), 790--812.

\bibitem[\protect\citeauthoryear{Newey and Powell}{Newey and Powell}{1987}]{newey1987asymmetric}
Newey, W.~K. and J.~L. Powell (1987).
\newblock Asymmetric least squares estimation and testing.
\newblock {\em Econometrica: Journal of the Econometric Society\/}, 819--847.

\bibitem[\protect\citeauthoryear{Powell}{Powell}{2020}]{powell2020quantile}
Powell, D. (2020).
\newblock Quantile treatment effects in the presence of covariates.
\newblock {\em Review of Economics and Statistics\/}~{\em 102\/}(5), 994--1005.

\bibitem[\protect\citeauthoryear{Racine}{Racine}{2000}]{racine2000consistent}
Racine, J. (2000).
\newblock Consistent cross-validatory model-selection for dependent data: hv-block cross-validation.
\newblock {\em Journal of econometrics\/}~{\em 99\/}(1), 39--61.

\bibitem[\protect\citeauthoryear{Rodrigues and Fan}{Rodrigues and Fan}{2017}]{rodrigues2017regression}
Rodrigues, T. and Y.~Fan (2017).
\newblock Regression adjustment for noncrossing bayesian quantile regression.
\newblock {\em Journal of Computational and Graphical Statistics\/}~{\em 26\/}(2), 275--284.

\bibitem[\protect\citeauthoryear{Ruzicka}{Ruzicka}{2021}]{ruzicka2021quantile}
Ruzicka, J. (2021).
\newblock Quantile local projections: Identification, smooth estimation, and inference.
\newblock {\em Universidad Carlos III de Madrid. True SQF Fitted linear SQF using QR\/}.

\bibitem[\protect\citeauthoryear{Shao}{Shao}{1997}]{shao1997asymptotic}
Shao, J. (1997).
\newblock An asymptotic theory for linear model selection.
\newblock {\em Statistica sinica\/}, 221--242.

\bibitem[\protect\citeauthoryear{Sobotka and Kneib}{Sobotka and Kneib}{2012}]{sobotka2012geoadditive}
Sobotka, F. and T.~Kneib (2012).
\newblock Geoadditive expectile regression.
\newblock {\em Computational Statistics \& Data Analysis\/}~{\em 56\/}(4), 755--767.

\bibitem[\protect\citeauthoryear{Stone}{Stone}{1977}]{stone1977consistent}
Stone, C.~J. (1977).
\newblock Consistent nonparametric regression.
\newblock {\em The annals of statistics\/}, 595--620.

\bibitem[\protect\citeauthoryear{Szendrei and Varga}{Szendrei and Varga}{2023}]{szendrei2023revisiting}
Szendrei, T. and K.~Varga (2023).
\newblock Revisiting vulnerable growth in the euro area: Identifying the role of financial conditions in the distribution.
\newblock {\em Economics Letters\/}, 110990.

\bibitem[\protect\citeauthoryear{Wager}{Wager}{2020}]{wager2020cross}
Wager, S. (2020).
\newblock Cross-validation, risk estimation, and model selection: Comment on a paper by rosset and tibshirani.
\newblock {\em Journal of the American Statistical Association\/}~{\em 115\/}(529), 157--160.

\bibitem[\protect\citeauthoryear{Wojciechowski}{Wojciechowski}{2024}]{wojciechowski2024structural}
Wojciechowski, R. (2024).
\newblock A structural approach to growth-at-risk.
\newblock {\em arXiv preprint arXiv:2410.04431\/}.

\bibitem[\protect\citeauthoryear{Yang}{Yang}{2005}]{yang2005can}
Yang, Y. (2005).
\newblock Can the strengths of aic and bic be shared? a conflict between model indentification and regression estimation.
\newblock {\em Biometrika\/}~{\em 92\/}(4), 937--950.

\bibitem[\protect\citeauthoryear{Yang}{Yang}{2007}]{yang2007consistency}
Yang, Y. (2007).
\newblock {Consistency of cross validation for comparing regression procedures}.
\newblock {\em The Annals of Statistics\/}~{\em 35\/}(6), 2450 -- 2473.

\bibitem[\protect\citeauthoryear{Yang and Tokdar}{Yang and Tokdar}{2017}]{yang2017joint}
Yang, Y. and S.~T. Tokdar (2017).
\newblock Joint estimation of quantile planes over arbitrary predictor spaces.
\newblock {\em Journal of the American Statistical Association\/}~{\em 112\/}(519), 1107--1120.

\bibitem[\protect\citeauthoryear{Zou}{Zou}{2006}]{zou2006adaptive}
Zou, H. (2006).
\newblock The adaptive lasso and its oracle properties.
\newblock {\em Journal of the American statistical association\/}~{\em 101\/}(476), 1418--1429.

\bibitem[\protect\citeauthoryear{Zou and Yuan}{Zou and Yuan}{2008}]{zou2008composite}
Zou, H. and M.~Yuan (2008).
\newblock Composite quantile regression and the oracle model selection theory.
\newblock {\em The Annals of Statistics\/}~{\em 36\/}(3), 1108--1126.

\end{thebibliography}

\pagebreak

\appendix
\section{Estimators}

\begin{itemize}
    \item Quantile regression \citep{koenker1978regression}
    \begin{equation} \nonumber
        \hat{\beta}_{QR}=\underset{\beta}{argmin}\sum^{Q}_{q=1}\sum^{T}_{t=1}\rho_q(y_{t}-x_t^T\beta_{\tau_q})
    \end{equation}
    \item Composite quantile regression \citep{koenker1984note,zou2008composite}
    \begin{equation} \nonumber
        \hat{\beta}_{CQR}=\underset{\beta}{argmin}\sum^{Q}_{q=1}\sum^{T}_{t=1}\rho_q(y_{t}-x_t^T\beta)
    \end{equation}
    \item Non-crossing qunatile regression \citep{bondell2010noncrossing}. Here, $z \in [0,1]$, i.e. it is the $x$ rescaled.
    \begin{equation} \nonumber
    \begin{split}
        \hat{\beta}_{BRW}&=\underset{\beta}{argmin}\sum^{Q}_{q=1}\sum^{T}_{t=1}\rho_q(y_{t}-z_t^T\beta_{\tau_q})\\
        &s.t.~\gamma_{0,\tau_{p}}\geq \sum^K_{k=1} \gamma_{k,\tau_q}^-
    \end{split}
    \end{equation}
    \item Fused shrinkage quantile regression \citep{jiang2013interquantile}
    \begin{equation} \nonumber
    \begin{split}
        \hat{\beta}_{JWB}&=\underset{\beta}{argmin}\sum^{Q}_{q=1}\sum^{T}_{t=1}\rho_q(y_{t}-x_t^T\beta_{\tau_q})\\
        &s.t.~k^* \geq \sum_{q=2}^{Q}\sum^K_{k=1}(\gamma_{k,\tau_q}^++\gamma_{k,\tau_q}^-)
    \end{split}
    \end{equation}

    \item Generalised Non-Crossing Quantile Regression
    \begin{equation} \nonumber
    \begin{split}
        \hat{\beta}_{GNCQR}&=\underset{\beta}{argmin}\sum^{Q}_{q=1}\sum^{T}_{t=1}\rho_q(y_{t}-x_t^T\beta_{\tau_q})\\
        &s.t.~\gamma_{0,\tau_{p}}+\sum^K_{k=1} \Big[ \Bar{X_k} - \alpha(\Bar{X_k} - min(X_k)) \Big]\gamma_{k,\tau_q}^+\geq \sum^K_{k=1} \Big[\Bar{X_k} + \alpha(max(X_k)-\Bar{X_k}) \Big] \gamma_{k,\tau_q}^-
    \end{split}
    \end{equation}
\end{itemize}

\pagebreak

\section{Bias-variance trade-off discussion}
GNCQR allows us to gradually enforce non-crossing constraints by inducing fused shrinkage. This naturally entails that there is some type of bias-variance trade-off. In essence, as we increase $\alpha$ we increase the amount of non-crossing we want to impose. By introducing non-crossing constraints we impose a bias on the quantile property (i.e. we no longer obtain the minimum of the tick-loss weighted $\ell_1$ residuals) and instead enforce monotonically increasing quantiles.\footnote{The bias thus introduced by non-crossing constraints is similar to biased bootstrap methods (data tilting, data sharpening, etc) proposed for regression under monotone or convexity constraints. The degree of bias can be used to test model adequacy, which is a topic retained for future research.} This implicitly enforces that the data below $\tau_q$ must be a subset of the data below $\tau_{q+1}$. As we move from $\alpha=0$ to $\alpha=1$, we gradually move from the quantile property to the `quantile subset' property. Once $\alpha>1$, GNCQR starts to become more restrictive in how the quantile subset property is satisfied. In particular, it will start to penalise quantile closeness as well as quantile crossing. Formally this can be shown as follows:

\begin{equation}\label{eq:quantileclose}
    \{ I(\varepsilon_{\tau_{q-1}}-\xi\leq0) \} \subseteq \{ I(\varepsilon_{\tau_{q}}\leq 0) \}
\end{equation}

Here $\{ I(\varepsilon_{\tau_{q}}\leq0) \}$ is the elements of the quantile residual that are negative for the fitted $\tau_q^{th}$ quantile, i.e. the observations below the fitted quantile. The quantile subset property necessitates that the observations with negative residuals of the $\tau_{q-1}$ quantile must be a subset of the observations with negative residuals of the $\tau_q$ quantile. This is achieved in equation (\ref{eq:quantileclose}) when $\xi=0$.

As $\alpha$ of equation (\ref{eq:GNCQR}) goes above 1, the $\xi$ term in the quantile subset property (equation (\ref{eq:quantileclose})) increases as well. Increasing the $\xi$ parameter means that even some observations that are above the $\tau_{q-1}$ fitted quantile form the set that has to be below $\tau_q$ quantile. In this way, quantile that get too close to each other are treated as if they cross. As such, when $\alpha>1$, the $\xi$ term in equation (\ref{eq:quantileclose}) becomes positive, which in turn leads to penalising quantile closeness in addition to quantile crossing. Conversely, when $\alpha<1$, the $\xi$ term is negative and the quantile subset property allows for some degree of quantile crossing. When $\alpha=0$, the $\xi$ term is some large number, making equation (\ref{eq:quantileclose}) trivial to satisfy.

\pagebreak

\section{Hyperparameter selection methods}
The choice of cross-validation technique can lead to differences in results. In particular, there is a documented trade-off between model identification and minimising predictive risk \citep{yang2005can}. Broadly speaking, leave-one-out cross-validation is asymptotically equivalent to AIC \citep{stone1977consistent}, while the various block cross-validation methods are closer to the BIC \citep{shao1997asymptotic}, if the size of the training sample (relative to the validation sample) goes to zero as $n\rightarrow\infty$. This underlies the key point of \citet{yang2005can} that one cannot simultaneously achieve model consistency (in terms of model selection) and efficiency (in terms of achieving lowest error variance). Note also that by choosing the $hv$-block CV of \citet{racine2000consistent}, we are implicitly expressing a preference in favour of model selection. If the interest is purely data fit, it might be beneficial to opt for a leave-one-out CV. We retain consideration of the GNCQR using other CV methods to future research.

Given the results of \citet{stone1977consistent} and \citet{shao1997asymptotic}, one can also rely on information criteria for hyperparameter selection. The equation for AIC and BIC for quantile regression can be found in \citet{jiang2014interquantile}. Utilising the information criteria has the advantage of only needing one estimation (per hyperparameter), rather than one per block. This can speed up computation especially for larger sample sizes.

Cross-validation has been shown to have some drawbacks for some applications \citep{bates2024cross}. Specifically, it has been shown that cross-validation estimates the average prediction error of models fit on other unseen training sets drawn from the same population, rather than the prediction error of the model fit on the specific training set. On account of this, confidence intervals for prediction error may have lower coverage. While these results are undoubtedly important, they are less of a concern for model selection. In essence, all candidate hyperparameters are biased in the same way which has little influence on their relative rank. This entails that the optimal hyperparameter choice remains valid.\footnote{For a more rigorous mathematical treatment of the consistency of cross-validation for model selection, see Theorem 2 of \citet{yang2007consistency} and Proposition 2 of \citet{wager2020cross}.}\textsuperscript{,} \footnote{While our analysis assumes that the relative ranking of hyperparameters remains valid, further investigation might reveal contexts in which Bates' nested cross-validation method improves hyperparameter tuning.}




\end{document}